\documentclass[10pt, onecolumn]{IEEEtran}
\usepackage{soul}
\usepackage[mathscr]{eucal}
\usepackage[cmex10]{amsmath}
\usepackage{epsfig,epsf,psfrag}
\usepackage{amssymb,amsmath,amsthm,amsfonts,latexsym}
\usepackage{amsmath,graphicx,bm,xcolor,url,overpic}
\usepackage{fixltx2e}
\usepackage{array}
\usepackage{verbatim}
\usepackage{bm}
\usepackage{algorithmic}
\usepackage{algorithm}
\usepackage{verbatim}
\usepackage{textcomp}
\usepackage{mathrsfs}
\usepackage{epstopdf}

\newcommand{\openone}{\leavevmode\hbox{\small1\normalsize\kern-.33em1}}

\catcode`~=11 \def\UrlSpecials{\do\~{\kern -.15em\lower .7ex\hbox{~}\kern .04em}} \catcode`~=13

\allowdisplaybreaks[4]

\newcommand{\nn}{\nonumber}

\newcommand{\calE}{\mathcal{E}}

\newcommand{\calL}{\mathcal{L}}

\newcommand{\calN}{\mathcal{N}}

\newcommand{\calP}{\mathcal{P}}
\newcommand{\calQ}{\mathcal{Q}}
\newcommand{\calR}{\mathcal{R}}

\newcommand{\calV}{\mathcal{V}}

\newcommand{\calX}{\mathcal{X}}

\newcommand{\by}{\mathbf{y}}
\newcommand{\bY}{\mathbf{Y}}

\newcommand{\rmA}{\mathrm{A}}

\newcommand{\rmd}{\mathrm{d}}

\newcommand{\rmP}{\mathrm{P}}

\newcommand{\rmQ}{\mathrm{Q}}

\newcommand{\rmV}{\mathrm{V}}

\newcommand{\bbE}{\mathsf{E}}

\newcommand{\bbN}{\mathbb{N}}

\newcommand{\bbR}{\mathbb{R}}

\DeclareMathAlphabet{\mathbsf}{OT1}{cmss}{bx}{n}
\DeclareMathAlphabet{\mathssf}{OT1}{cmss}{m}{sl}

\DeclareSymbolFont{bsfletters}{OT1}{cmss}{bx}{n}
\DeclareSymbolFont{ssfletters}{OT1}{cmss}{m}{n}
\DeclareMathSymbol{\bsfGamma}{0}{bsfletters}{'000}
\DeclareMathSymbol{\ssfGamma}{0}{ssfletters}{'000}
\DeclareMathSymbol{\bsfDelta}{0}{bsfletters}{'001}
\DeclareMathSymbol{\ssfDelta}{0}{ssfletters}{'001}
\DeclareMathSymbol{\bsfTheta}{0}{bsfletters}{'002}
\DeclareMathSymbol{\ssfTheta}{0}{ssfletters}{'002}
\DeclareMathSymbol{\bsfLambda}{0}{bsfletters}{'003}
\DeclareMathSymbol{\ssfLambda}{0}{ssfletters}{'003}
\DeclareMathSymbol{\bsfXi}{0}{bsfletters}{'004}
\DeclareMathSymbol{\ssfXi}{0}{ssfletters}{'004}
\DeclareMathSymbol{\bsfPi}{0}{bsfletters}{'005}
\DeclareMathSymbol{\ssfPi}{0}{ssfletters}{'005}
\DeclareMathSymbol{\bsfSigma}{0}{bsfletters}{'006}
\DeclareMathSymbol{\ssfSigma}{0}{ssfletters}{'006}
\DeclareMathSymbol{\bsfUpsilon}{0}{bsfletters}{'007}
\DeclareMathSymbol{\ssfUpsilon}{0}{ssfletters}{'007}
\DeclareMathSymbol{\bsfPhi}{0}{bsfletters}{'010}
\DeclareMathSymbol{\ssfPhi}{0}{ssfletters}{'010}
\DeclareMathSymbol{\bsfPsi}{0}{bsfletters}{'011}
\DeclareMathSymbol{\ssfPsi}{0}{ssfletters}{'011}
\DeclareMathSymbol{\bsfOmega}{0}{bsfletters}{'012}
\DeclareMathSymbol{\ssfOmega}{0}{ssfletters}{'012}

\newcommand{\hatx}{\hat{x}}
\newcommand{\hatX}{\hat{X}}

\DeclareMathOperator*{\argmin}{arg\,min}

\newcommand{\bzero}{\mathbf{0}}

\newtheorem{theorem}{Theorem}
\newtheorem{lemma}[theorem]{Lemma}

\newtheorem{definition}{Definition}

\graphicspath{{./figures/}}
\usepackage{graphicx,cite}
\usepackage{epstopdf}
\usepackage{enumerate}
\usepackage[ colorlinks = true,
             linkcolor = blue,
             urlcolor  = blue,
             citecolor = red,
             anchorcolor = green,]{hyperref}
\usepackage{soul,color}
\usepackage{subfigure}

\newcommand{\myfoot}[1]{\footnote{\color{red}\bf #1}}
\soulregister\em7
\soulregister\cite7
\soulregister\ref7
\soulregister\eqref7
\soulregister\underline7
\soulregister\emph7
\soulregister\footnote7
\soulregister\myfoot7
\definecolor{Dyellow}{RGB}{254,152,0}
\definecolor{Dgreen}{RGB}{0,176,80}
\newcommand{\red}[1]{\textcolor{red}{#1}}

\begin{document}
\title{Achievable Refined Asymptotics for Successive Refinement Using Gaussian Codebooks}

\author{Lin Bai, Zhuangfei~Wu, and~Lin~Zhou

\thanks{This work has been partially presented at ISIT 2021~\cite{wu2021} and ISIT 2022~\cite{wu2022ISIT}.}
\thanks{The authors are with the School of Cyber Science and Technology, Beihang University, Beijing, China, 100191 (Emails: \{l.bai,zhuangfeiwu,lzhou\}@buaa.edu.cn). L. Bai and L. Zhou are also with the Beijing Laboratory for General Aviation Technology, Beihang University, Beijing.}}

\maketitle
\allowdisplaybreaks[1]

\begin{abstract}
We study the mismatched successive refinement problem where one uses Gaussian codebooks to compress an arbitrary memoryless source with successive minimum Euclidean distance encoding under the quadratic distortion measure. Specifically, we derive achievable refined asymptotics under both the joint excess-distortion probability (JEP) and the separate excess-distortion probabilities (SEP) criteria. For both second-order and moderate deviations asymptotics, we consider two types of codebooks: the spherical codebook where each codeword is drawn independently and uniformly from the surface of a sphere and the i.i.d. Gaussian codebook where each component of each codeword is drawn independently from a Gaussian distribution. We establish the achievable second-order rate-region under JEP and we show that under SEP any memoryless source satisfying mild moment conditions is strongly successively refinable. When specialized to a Gaussian memoryless source (GMS), our results provide an alternative achievability proof with specific code design. We show that under JEP and SEP, the same moderate deviations constant is achievable. For large deviations asymptotics, we only consider the i.i.d. Gaussian codebook since the i.i.d. Gaussian codebook has better performance than the spherical codebook in this regime for the one layer mismatched rate-distortion problem~(Zhou, Tan, Motani, TIT, 2019). We derive achievable exponents of both JEP and SEP and specialize our results to a GMS, which appears to be a novel result of independent interest.
\end{abstract}

\section{Introduction}
The successive refinement problem~\cite{rimoldi1994} was motivated by diverse applications such as image and video compression and clinical diagnosis using X-rays. As shown in Fig. \ref{fig:systemmodel}, in this problem, there are two layers of encoders and decoders. For each $i\in\{1,2\}$, the encoder $f_i$ has access to a source sequence $X^n$ and compresses it into a message $S_i$. Two decoders aim to recover the source sequence with different distortion requirements and different access to compressed information. Specifically, decoder $\phi_1$ aims to reproduce the source sequence $X^n$ within distortion level $D_1$ using the compressed information $S_1$ from encoder $f_1$ and the decoder $\phi_2$ aims to reproduce $X^n$ within a finer distortion level $D_2<D_1$ with the additional access to the compressed information $S_2$ from encoder $f_2$. Let the reproduced versions at decoders $\phi_1$ and $\phi_2$ be $\hatX_1^n$ and $\hatX_2^n$, respectively.

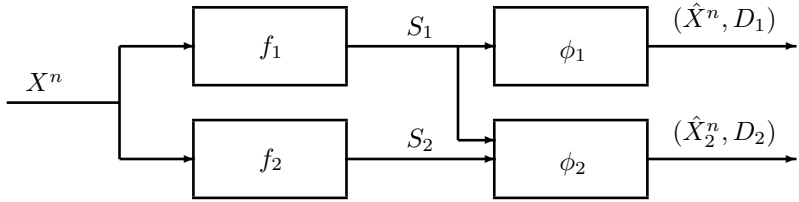
\begin{figure}[t]
\centering
\setlength{\unitlength}{0.5cm}
\scalebox{1}{
\begin{picture}(26,6)
\linethickness{1pt}
\put(1.5,3.8){\makebox{$X^n$}}
\put(6,1){\framebox(4,2)}
\put(6,4){\framebox(4,2)}
\put(7.7,1.8){\makebox{$f_2$}}
\put(7.7,4.8){\makebox{$f_1$}}
\put(1,3.5){\line(1,0){3}}
\put(4,5){\vector(1,0){2}}
\put(4,2){\line(0,1){3}}
\put(4,2){\vector(1,0){2}}
\put(14,1){\framebox(4,2)}
\put(14,4){\framebox(4,2)}

\put(15.7,4.7){\makebox{$\phi_1$}}
\put(15.7,1.7){\makebox{$\phi_2$}}
\put(10,2){\vector(1,0){4}}
\put(12,2.5){\makebox(0,0){$S_2$}}
\put(10,5){\vector(1,0){4}}
\put(12,5.5){\makebox(0,0){$S_1$}}
\put(13,5){\line(0,-1){2.5}}
\put(13,2.5){\vector(1,0){1}}
\put(18,2){\vector(1,0){4}}
\put(18.7,2.5){\makebox{$(\hatX_2^n,D_2)$}}
\put(18,5){\vector(1,0){4}}
\put(18.7,5.5){\makebox{$(\hatX^n,D_1)$}}
\end{picture}
}
\caption{System Model for Successive Refinement.}
\label{fig:systemmodel}
\end{figure}

To evaluate the performance of a coding scheme, there are two different performance criteria for successive refinement:
the joint excess-distortion probability (JEP) that measures the probability that either decoder fails to decode $X^n$ within the desired distortion level and the separate excess-distortion probabilities (SEP) that measure the decoding failure of two decoders separately. For any discrete memoryless source (DMS), Rimoldi~\cite{rimoldi1994} derived the rate-distortion region that asymptotically characterizes the rate requirements of both encoders to ensure reliable recovery at decoders with vanishing JEP. Subsequently, Kanlis and Narayan~\cite{kanlis1996error} refined Rimoldi's result by showing that for any rate pair strictly inside the rate-distortion region, the JEP decays exponentially fast to zero and characterized the optimal decay rate. Koshelev~\cite{koshelev1981estimation} and Equitz-Cover~\cite{equitz1991successive} studied the so-called successive refinability property which shows that the rate-distortion region can be reduced to the case where rates of both encoders are bounded by rates of a single layer rate-distortion function with desired distortion levels under certain conditions on the source distribution and the distortion measures. For example, a Gaussian memoryless source (GMS) under the quadratic distortion measure is successively refinable~\cite{equitz1991successive}. Recently, Zhou, Tan and Motani~\cite{zhou2016second} refined Rimoldi's results by deriving second-order and moderate deviations asymptotics for any DMS and GMS. Under SEP, Tuncel and Rose~\cite{tuncel2003} derived error exponents for any DMS, which implied that the rate-distortion region remains the same. The authors of \cite{tuncel2003} also discovered an interesting tradeoff between the exponents of SEP and identified the conditions under which the exponents of the two layers are successively refinable. Under SEP, No, Ingber and Weissman~\cite{no2016} derived achievable second-order asymptotics for both a DMS and a GMS and extended the successive refinability property to the stronger second-order asymptotic case.

The above studies of successive refinement are very insightful. However, all these works assume that the distribution of the source sequence $X^n$ is perfectly known, which we term as the matched case. Such an assumption is impractical since usually  it is challenging to obtain the exact source distribution before compressing a source. Thus, one should use coding schemes ignorant of the source distribution. This direction was pioneered by Lapidoth~\cite{lapidoth1997} who proposed to compress an arbitrary memoryless source using i.i.d. Gaussian codebooks with minimum Euclidean distance encoding under the quadratic distortion measure for the rate-distortion problem, i.e.,the first layer of successive refinement in Fig. \ref{fig:systemmodel}. Specifically, Lapidoth~\cite[Theorem 3]{lapidoth1997} showed that for any ergodic source with finite second moment $\sigma^2$, the minimal compression rate that guarantees a vanishing excess-distortion probability is exactly the rate-distortion function for a GMS. Recently, Zhou, Tan and Motani \cite{zhou2018refined} refined Lapidoth's results by deriving ensemble tight refined asymptotics for the same setting using both i.i.d. Gaussian codebooks and spherical codebooks.

One might wonder whether it is possible to generalize the the results in \cite{lapidoth1997,zhou2018refined} to the successive refinement problem. Specifically, can one propose a coding scheme to compress an arbitrary memoryless source using Gaussian codebooks for the successive refinement problem and ensure universal good performance under the quadratic distortion measure? In this paper, we answer this question affirmatively. Our main contributions are summarized as follows.

\subsection{Main Contributions}
We propose to use Gaussian codebooks to compress an arbitrary memoryless source with successive minimum Euclidean distance encoding under the quadratic distortion measure. Specifically, we consider two types of codebooks: the spherical codebook and the i.i.d. Gaussian codebook. The codewords of a spherical codebook are generated independently and uniformly from the surface of a sphere while each codeword of an i.i.d. Gaussian codebook is drawn independently from a product Gaussian distribution. We consider two types of performance criterion: ensemble JEP and ensemble SEP. Similar to \cite{lapidoth1997,zhou2018refined}, by ``ensemble'', we mean that the probability is calculated not only with respect to the distribution of the source but also over the random codebooks. Under the above setting, we derive achievable refined asymptotics under both ensemble JEP and SEP performance criteria.

We first derive achievable second-order and moderate deviations asymptotics for both spherical and i.i.d. Gaussian codebooks. Our results complement existing literature~\cite{zhou2018refined,zhou2016second,no2016,lapidoth1997}. Specifically, our results generalize the refined asymptotics analyses of the mismatched rate-distortion problem~\cite{zhou2018refined} to the more complicated successive refinement setting. Similar to \cite{zhou2018refined}, we find that both spherical and i.i.d. Gaussian codebooks achieve the same second-order and moderate deviations asymptotics under both JEP and SEP. By letting the blocklength tend to infinity, a corollary of our result states that asymptotically the achievable rate-distortion region for any memoryless source under our mismatched coding scheme equals the rate-distortion region of the GMS in the matched successive refinement problem, which generalizes the classical result of Lapidoth~\cite[Theorem 3]{lapidoth1997} to the successive refinement setting. Furthermore, our results generalize~\cite{no2016} and \cite{zhou2016second} for the matched successive refinement problem to the mismatched case. Under SEP, we generalize the concept ``strongly successively refinable~\cite[Definition 4]{no2016}'' from the matched case with perfect knowledge of the source distribution to the mismatched case. We find that any memoryless source is strongly successively refinable under mild conditions using our mismatched coding scheme (cf. Theorem~\ref{theo:second_SEP}). Finally, when specializing to a GMS, our results for second-order asymptotics under JEP and SEP provide alternative achievability proofs with explicit code design for the results in \cite[Theorem 20]{zhou2016second} and \cite[Theorem 7]{no2016}, respectively.

We next derive achievable large deviations asymptotics for the i.i.d. Gaussian codebook. We do not consider spherical codebook in this regime because Zhou, Tan and Motani~\cite[Lemma 4]{zhou2018refined} showed that i.i.d. Gaussian codebook has strictly larger excess-distortion exponent than the spherical codebook for the mismatched rate-distortion problem. We derive the achievable excess-distortion exponents of both JEP and SEP and show the exponents are all positive for any rate pair strictly inside the rate-distortion region of a GMS in the matched case. Note that in our conference version~\cite{wu2022ISIT}, we used a mismatched coding scheme where the power of each codeword is invariant of the coding rates and not able to ensure positive exponent under either SEP or JEP for all rate pairs inside the rate-distortion region. In this paper, we manage to solve the problem by slightly changing the design of the mismatched coding scheme and allowing the number of codewords of both encoders to be a function of their rates (see Section \ref{sec:largedeviations} for detailed discussions.) To our best knowledge, our result is the first achievable large deviations analysis for a continuous memoryless source for the successive refinement problem, whose matched counterpart is not even available for a GMS. Furthermore, our results generalize \cite[Theorem 3]{zhou2018refined} for the i.i.d. Gaussian codebook to the successive refinement setting and generalize \cite{kanlis1996error,tuncel2003} to the mismatched case.

\subsection{Other Related Works}
We briefly summarize other works on matched successive refinement and mismatched lossy source coding. Effros~\cite{effros1999} generalized Rimoldi's characterization of the rate-distortion region~\cite{rimoldi1994} to the case of discrete stationary ergodic and non-ergodic sources. Kostina and Tuncel~\cite{Kostina2019} studied successive refinement for abstract sources and derived non-asymptotic converse bounds that are tight for second-order asymptotics under both JEP and SEP criteria. Dembo and Kontoyiannis~\cite{Dembo2002} derived the asymptotic coding rate for mismatched compression using a lossy version of the asymptotic equipartition property. Kontoyiannis and Zamir~\cite{Kontoyiannis2006} applied variable-length source coding, also known as entropy coding, to mismatched compression and proved that asymptotically the rate-distortion function can be achieved with a slight loss for any source with unknown distribution. Zhou, Tan and Motani~\cite{zhou2019jscc} proposed a mismatched joint source channel coding scheme to transmit an arbitrary memoryless source over an additive arbitrary noise channel and derived the ensemble tight results in second-order and moderate deviation asymptotics for both spherical codebooks and i.i.d. Gaussian codebooks. Different from the mismatched setting considered in this paper where the source distribution is unknown, Lapidoth also considered another case of the mismatched setting for lossy source coding, where the encoder and decoder use different distortion measures~\cite[Theorem 1]{lapidoth1997}. Kanabar and Scarlett~\cite{kanabar2022mismatched} revisited the problem in ~\cite[Theorem 1]{lapidoth1997} and derived the achievable rate-distortion function using superposition coding.

\subsection{Organization of the Rest of the Paper}
The rest of the paper is organized as follows. In Section \ref{sec:problem_fomulation}, we set up the notation, formulate the problem and present necessary definitions. In Section \ref{sec:main_results}, we present our achievable refined asymptotic results under both JEP and SEP and discuss the significance of our results. The proofs of our results are presented in Sections \ref{sec:proof_SOJEP} to \ref{sec:proof_large}. Finally, in Section \ref{sec:conclusion}, we conclude the paper and discuss future research directions. For smooth presentation, we defer the proofs of all supporting lemmas to the appendices.

\section{Problem Formulation}
\label{sec:problem_fomulation}
\subsection*{Notation}
Random variables are in capital (e.g., $X$) and their realizations are in lower case (e.g., $x$). Random vectors of length $n$ and their particular realizations are denoted as $X^n:= (X_1, \ldots, X_n)$ and $x^n=(x_1,\ldots,x_n)$, respectively. We use calligraphic font (e.g., $\mathcal{X}$) to denote all sets.
We use $\bbR$, $\bbR_+$, $\bbN$ to denote the set of real numbers, positive real numbers and integers respectively. For any two integers $(a,b)\in\bbN^2$, we use $[a:b]$ to denote the set of integers between $a$ and $b$, and we use $[a]$ to denote $[1:a]$. We use logarithm with base $e$. Gaussian complementary cumulative distribution function (cdf) and its inverse are denoted as $Q(\cdot)$ and $Q^{-1}(\cdot)$, respectively. We use $\| x^n \| = \sqrt {\sum\nolimits_i {x_i^2} } $ to denote the $\ell_2$ norm of a vector $x^n \in \mathbb{R}^n$. The quadratic distortion measure for any two sequences $(x^n,y^n)\in(\bbR^n)^2$ is defined as $d(x^n,y^n):= \frac{1}{n}\| {{x^n} - {y^n}} \|^2 = \frac{1}{n}\sum_{i\in[n]} ( x_i - y_i )^2 $. We use $1\{\cdot\}$ to denote the indicator function, i.e., $1\{\rmA\}=1$ if A is true and otherwise $1\{\rmA\}=0$. For two positive sequences $\{a_n\}$ and $\{b_n\}$, we write $a_n\sim b_n$ if $\lim_{n\to\infty}\frac{1}{n}\log\frac{a_n}{b_n}=0$. Given any real number $t\in\bbR$, we use $\Lambda_X(t)$ to denote its cumulant generating function and use $\Lambda_{X}^*(t)$ to denote the Fenchel-Legendre transform of the cumulant generating function of the random variable $X$, i.e., $\Lambda_X(t)=\log\bbE[\exp(tX)]$ and $\Lambda_X^*(t)=\sup_{\theta\geq 0}\{\theta t-\Lambda_X(\theta)\}$. Finally, we use $|a|^+$ to denote $\max\{0,a\}$.

\subsection{Problem Formulation}

Consider a memoryless source  with distribution $P_X$ \footnote{The distribution is either discrete or continuous. For a discrete source, $P_X$ is a probability mass function. For a continuous source, $P_X$ is a probability density function.} defined on an alphabet $\calX$ satisfying the moment constraint
\begin{align}
\bbE[X^2]=\sigma^2. \label{def:arbitrary_source}
\end{align}
In this paper, we generalize the rate-distortion saddle-point problem~\cite[Theorem 3]{lapidoth1997} to the successive refinement setting~\cite{rimoldi1994}. Specifically, we study the successive refinement problem where one is constrained to use Gaussian codebooks (either spherical or i.i.d.) and the minimal Euclidean distance encoding to compress an arbitrary memoryless source $X^n$ under the quadratic distortion measure. Consistent with \cite{lapidoth1997,zhou2018refined}, we name our problem as mismatched successive refinement. Without loss of generality, we assume that the desired distortion levels at the two decoders satisfy $\sigma^2>D_1>D_2$. A definition of a code is as follows.

\begin{definition}
\label{def:srmismatch}
An $(n,M_1,M_2)$-code for the mismatched successive refinement problem consists of:
\begin{itemize}
\item a set of $M_1$ codewords $ \{ Y^n(i)  \}_{i\in[M_1]}$ known by encoders $(f_1,f_2)$ and decoders $(\phi_1,\phi_2)$,
\item a set of $M_2$ codewords $\{Z^n(i,j)\}_{j\in[M_2]}$ known by the encoder $f_2$ and the decoder $\phi_2$ for each $i\in[M_1]$,
\item two encoders $f_1$ and $f_2$ that use successive minimum Euclidean distance encoding to compresses the source sequence $X^n$, i.e.,
\begin{align}
f_1( X^n )&: = \argmin_{i \in [M_1]} d( X^n,Y^n(i) ),\\
f_2( X^n )&: = \argmin_{j \in [M_2]} d( X^n,Z^n( f_1(X^n),j )),
\end{align}
\item two decoders $\phi_1$ and $\phi_2$ that operate as follows:
\begin{align}
\phi_1 (f_1(X^n))&= Y^n(f_1(X^n)),\\
\phi_2(f_1(X^n),f_2(X^n))&= Z^n(f_1(X^n),f_2(X^n)).
\end{align}
\end{itemize}
\end{definition}
For simplicity, throughout the paper, we use $\hatX_1^n$ and $Y^n(f_1(X^n))$ interchangeably to denote the reproduced source sequence of the first decoder $\phi_1$, and similarly we use $\hatX_2^n$ and $Z^n(f_1(X^n),f_2(X^n))$ interchangeably.

To specify the codewords in Definition \ref{def:srmismatch}, we need to define the following distributions that are  functions of the two parameters: a vector  $c^n=(c_1,\ldots,c_n)\in\bbR^n$ and a positive real number $P\in\bbR_+$.
\begin{itemize}
\item We first define a uniform distribution over the surface of a sphere with center $c^n$ and radius $\sqrt{nP}$, i.e. for any $u^n=(u_1,\ldots,u_n)\in\bbR^n$,
\begin{align}
f_{\rm{sp}}(u^n|c^n,P)
= \frac{1\{ \| u^n - c^n\|^2 - nP\}}{S_n(\sqrt{nP})},
\end{align}
where $1\{\cdot\}$ is the indicator function, $S_n(r)=n\pi^{n/2}r^{n-1}/\Gamma( {\frac{n+2}{2}})$ is the surface area of an $n$-dimensional
 sphere with radius $r$, and $\Gamma(\cdot)$ is the Gamma function.
\item  We also define a product Gaussian distribution, i.e., for any $u^n=(u_1,\ldots,u_n)\in\bbR^n$,
\begin{align}
f_{\rm{iid}}(u^n|c^n,P)
=\prod_{i\in[n]}\frac{\exp\left(-\frac{(u_i-c_i)^2}{2P}\right)}{\sqrt{2\pi P}}.\label{def:iidcodewords}
\end{align}
\end{itemize}

Next we specify the codebooks used in our mismatched coding scheme. Let $\lambda\in\bbR_+$ be a design parameter to be specified, $P_Y:=\sigma^2-\lambda D_1$ and $P_Z:=\lambda D_1-D_2$. Similar to \cite{zhou2018refined}, we consider both spherical and i.i.d. Gaussian codebooks.  Since both encoders can use either spherical or i.i.d. Gaussian codebooks, there are in total four different combinations of the codebooks. For ease of notation, throughout the paper, we use $\dagger$ to denote the type of the codebook used by encoder $f_1$ and $\ddagger$ to denote the type of the codebook used by encoder $f_2$. Note that $(\dagger,\ddagger)\in\rm\{sp,iid\}^2$ where ``sp'' is short  for ``spherical''. Specifically, for any $(\dagger,\ddagger)\in\rm\{sp,iid\}^2$, encoder $f_1$ uses a random codebook with $M_1$ independent codewords $(Y^n(1),\ldots,Y^n(M_1))$, where each codeword $Y^n(i)$ is generated according to the distribution $f_\dagger(Y^n(i)|\bzero^n,P_Y)$ where $\bzero^n$ denotes the length-$n$ vector with all elements of $0$; for each $i\in[M_1]$ and $Y^n(i)$, encoder $f_2$ uses a random codebook with $M_2$ independent codewords $(Z^n(i,1),\ldots,Z^n(i,M_2))$, where each codeword $Z^n(i,j)$ is generated according to the distribution $f_\ddagger(Z^n(i,j)|Y^n(i),P_Z)$. For $n=2$, we illustrate encoders' codebooks in Fig. \ref{fig:codewords}.
\begin{figure}[htbp]
\centering
\subfigure[Both encoders use spherical codebooks]{
\begin{minipage}[t]{0.45\linewidth}
\centering
\includegraphics[width=\columnwidth]{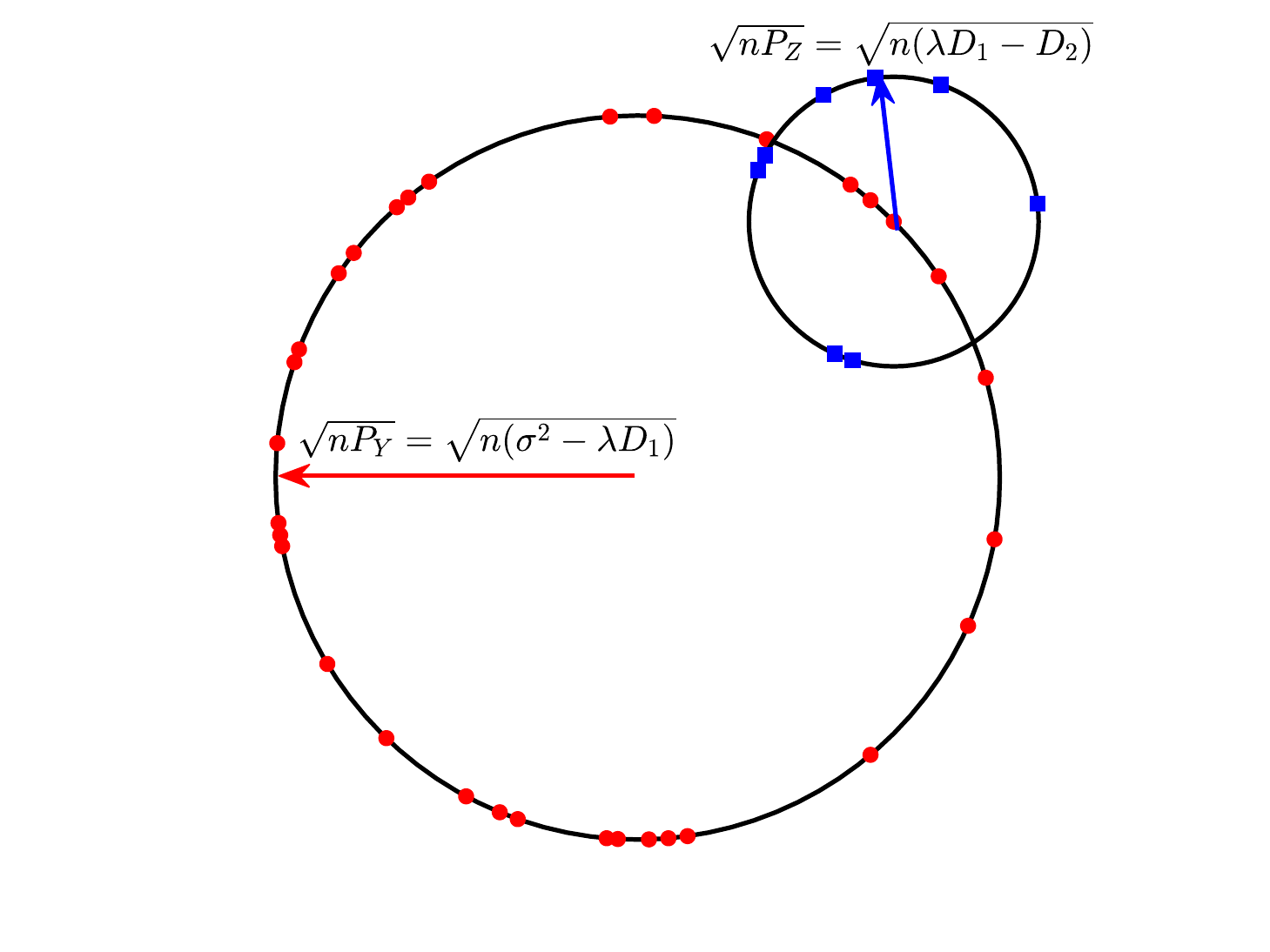}
\end{minipage}
}
\subfigure[Both encoders use i.i.d. Gaussian codebooks]{
\begin{minipage}[t]{0.45\linewidth}
\centering
\includegraphics[width=\columnwidth]{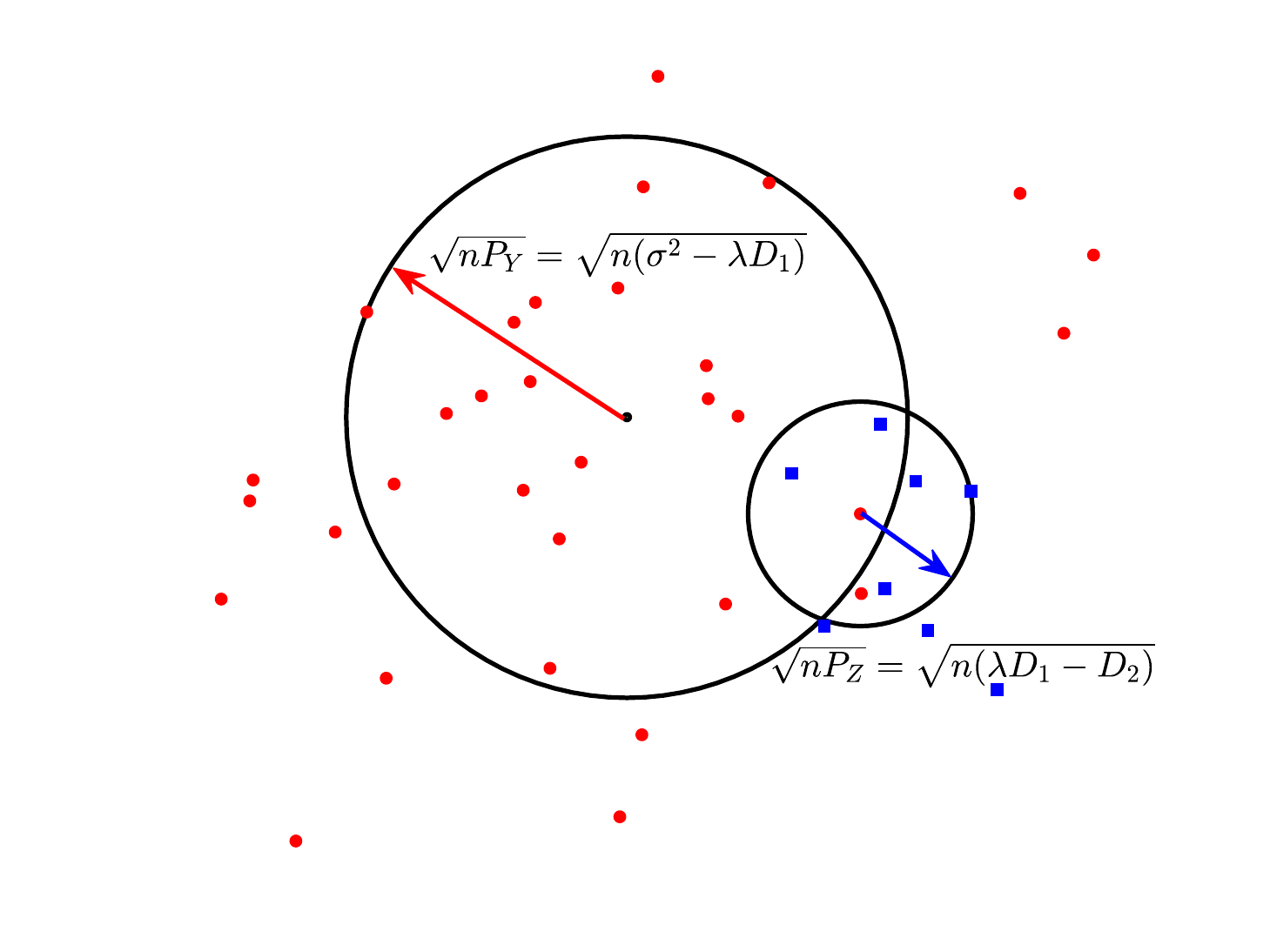}
\end{minipage}
}
\caption{
Illustration of codewords when both encoders use the same kind of codebooks. The red dots denote the codewords of encoder $f_1$ and the blue squares denote the codewords of encoder $f_2$. In a spherical codebook, the codewords of $f_1$ distribute uniformly over the surface of the sphere with center of the origin and the radius of $\sqrt{nP_Y}=\sqrt{n(\sigma^2-\lambda D_1)}$. When $f_1(X^n)=i$, given the codeword $Y^n(i)$, the codewords for $f_2$ distribute uniformly over the surface of the sphere with center $Y^n(i)$ and the radius $\sqrt{nP_Z}=\sqrt{n(\lambda D_1-D_2)}$. In an i.i.d. Gaussian codebook, the codewords of $f_1$ are generated independently and each codeword is generated i.i.d. from the Gaussian distribution with mean $0$ and variance $\sigma^2-\lambda D_1$. When $f_1(X^n)=i$, given the codeword $Y^n(i)$, the codewords for encoder $f_2$ are generated similarly from a Gaussian product distribution.
}
\label{fig:codewords}
\end{figure}

To evaluate the performance of the mismatched coding scheme, we consider the ensemble separate excess-distortion probabilities:
\begin{align}
\rmP_{\dagger}^n(D_1|M_1,M_2)&:=\Pr\{d(X^n,\hat{X}_1^n)>D_1\},\label{def:P_dagger}\\
\rmP_{\ddagger}^n(D_2|M_1,M_2)&:=\Pr\{d(X^n,\hat{X}_2^n)>D_2\},\label{def:P_ddagger}
\end{align}
and the ensemble joint-excess-distortion probability:
\begin{align}
\rmP_{\dagger,\ddagger}^n(D_1,D_2|M_1,M_2):=\Pr\{d(X^n,\hatX_1^n)>D_1\;\mathrm{or}\;d(X^n,\hatX_2^n)>D_2 \}.\label{def:P_joint}
\end{align}
Note that the probability terms average not only over the distribution of the source sequence $X^n$, but also over the random codebooks. These definitions are consistent with existing works on mismatched communication~\cite{lapidoth1996,lapidoth1997,zhou2018refined,scarlett2017mismatch}.

\subsection{The Rate-distortion Region}
The rate-distortion region collects rate pairs of both encoders with which our mismatched code in Definition \ref{def:srmismatch} ensures vanishing ensemble joint excess-distortion probability and vanishing ensemble separate excess-resolution probabilities. In the following, we define the rate-distortion region under JEP.
\begin{definition}
A rate pair $(R_1,R_2)\in\bbR_+^2$ is said to be $(D_1,D_2|\dagger,\ddagger)$-achievable for mismatched successive refinement if there exists a sequence of $(n,M_1,M_2)$-codes using $(\dagger,\ddagger)$ codebooks such that
\begin{align}
\limsup_{n\to\infty}\frac{1}{n}\log M_1&\leq R_1,\\
\limsup_{n\to\infty}\frac{1}{n}\log M_1M_2&\leq R_1+R_2,
\end{align}
and
\begin{align}
\lim_{n\to\infty}\rmP_{\dagger,\ddagger}^n(D_1,D_2|M_1,M_2)=0.
\end{align}
The convex closure of the set of all $(D_1,D_2|\dagger,\ddagger)$-achievable rate pairs is called the rate-distortion region and denoted as $\calR_{\dagger,\ddagger}(D_1,D_2)$.
\end{definition}

As a corollary of our result in Theorem \ref{theo:second_JEP} by letting $n\to\infty$, one has the following inner bound to $\calR_{\dagger,\ddagger}(D_1,D_2)$:
\begin{align}
\calR_{\mathrm{inner}}:=&\left\{(R_1,R_2):~R_1\geq \frac{1}{2}\log\frac{\sigma^2}{D_1},\;R_1+R_2\geq \frac{1}{2}\log\frac{\sigma^2}{D_2}\right\}\subseteq\calR_{\dagger,\ddagger}(D_1,D_2)\label{ach:rdregion}.
\end{align}

Since $\rmP_{\dagger,\ddagger}^n(D_1,D_2|M_1,M_2)\geq \max_{i\in[2]}\rmP_{\dagger}^n(D_i|M_1,M_2)$, the same inner bound holds under SEP. Note that $\calR_{\mathrm{inner}}$ is the rate-distortion region for a GMS under the quadratic distortion measure in the matched successive refinement problem~\cite{equitz1991successive,zhou2016second}. One might wonder whether the ensemble converse in~\cite[Theorem 3]{lapidoth1997} implies the ensemble converse for mismatched successive refinement. Unfortunately, the answer doesn't directly follow since the codewords for the second encoder of mismatched successive refinement are not equivalent to the codewords for the encoder in mismatched rate-distortion, a figure illustration of which are available Fig. \ref{fig:spsp_point-to-point} and Fig. \ref{fig:iidiid_can't_imply} where both encoders use either spherical or i.i.d. Gaussian codebooks. It is worthwhile future work to investigate whether the rate-distortion region $\calR_{\mathrm{inner}}$ is ensemble tight.

\begin{figure}[tb]
\centering
\includegraphics[width=0.5\columnwidth]{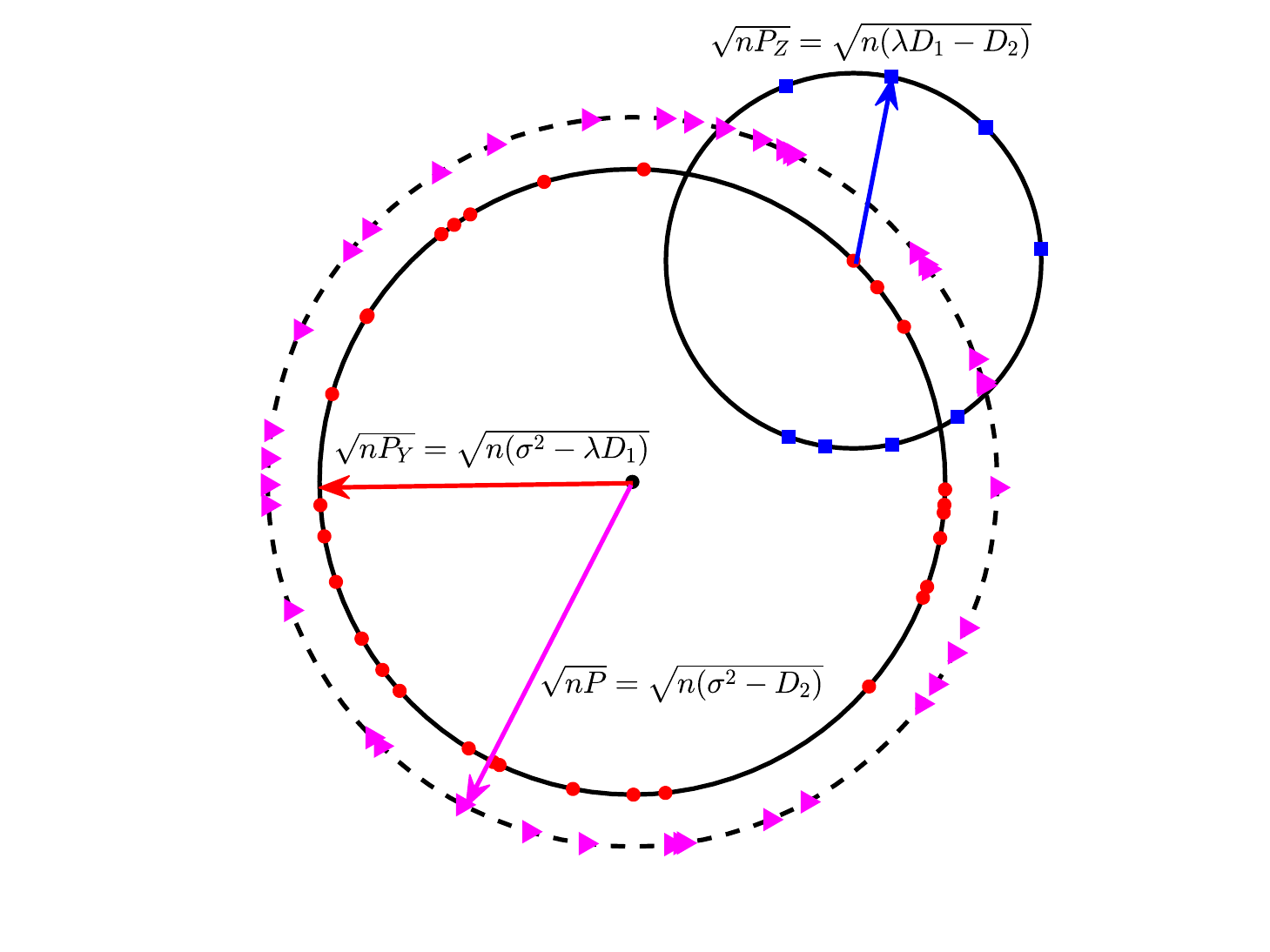}
\caption{Illustration of the spherical codebooks for $n=2$. The magenta triangle denote the codeword of encoder $f$, the red dots denote the codewords of encoder $f_1$ and the blue squares denote the codewords of encoder $f_2$. In point-to-point source coding, the codewords are generated uniformly over the surface of the sphere denoted as dash circle with center of origin and the radius $\sqrt{nP}=\sqrt{n(\sigma^2-D_2)}$. However, in successive refinement setting, the codewords of encoder $f_1$ are generated uniformly over the surface of the sphere with center of origin and the radius $\sqrt{nP_Y}=\sqrt{n(\sigma^2-\lambda D_1)}$ and the codewords of encoder $f_2$ are generated uniformly over the surface of the sphere centered on the given codeword of encoder $f_1$ with radius $\sqrt{nP_Z}=\sqrt{n(\lambda D_1-D_2)}$. Thus, the codewords of encoder $f_2$ are not uniformly distributed over the surface of the sphere, which makes it different from the spherical codebook used in the mismatched rate-distortion problem.}
\label{fig:spsp_point-to-point}
\end{figure}

\begin{figure}[tb]
\centering
\includegraphics[width=0.6\columnwidth]{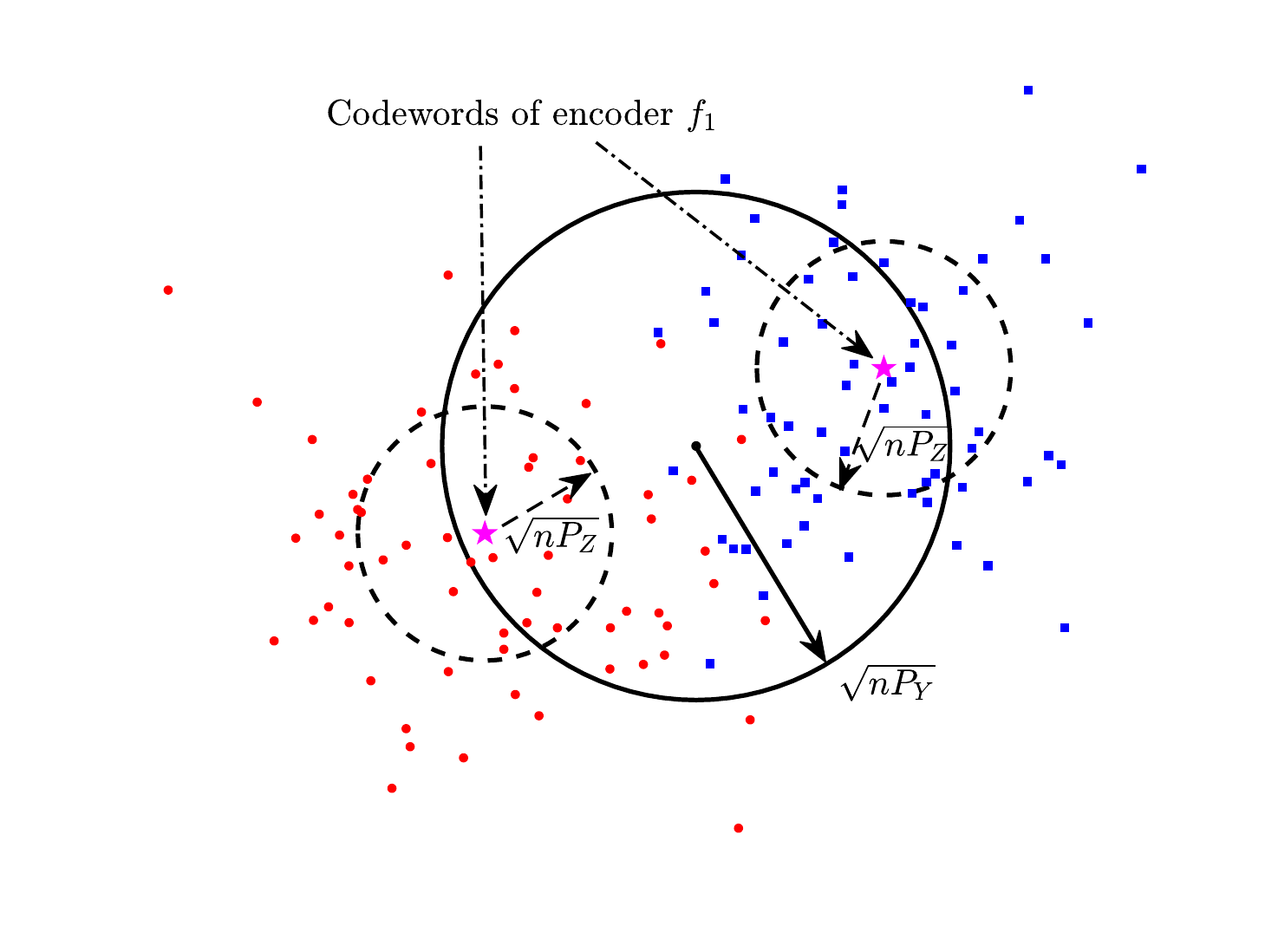}
\caption{Illustration of i.i.d. Gaussian codebooks for $n=2$. The magenta stars denote the codewords of encoder $f_1$ and the red dots and blue squares denote the codewords of encoder $f_2$. Assume that $M_1=2$. The codewords $Y^n(1),Y^n(2)$ of encoder $f_1$ are generated i.i.d. from the Gaussian distribution with mean $0$ and variance $P_Y$. Given each codeword $Y^n(i)$, the codewords $\{Z^n(i,j)\}_{j\in[M_2]}$ of encoder $f_2$ are generated i.i.d. from the Gaussian distribution with mean $Y^n(i)$ and variance $P_Z$. Although the marginal distribution of each codeword $Z^n(i,j)$ is Gaussian with mean $0$ and variance $P_Y+P_Z=\sigma^2-D_2$, the dependence of the codewords make it different from the i.i.d. Gaussian codebook in the mismatched rate-distortion problem.}
\label{fig:iidiid_can't_imply}
\end{figure}

\subsection{Definitions of Fundamental Limits}
In this paper, we interested in refined asymptotics that include the second-order, moderate and large deviations. These analyses reveal the tradeoff between the rates of encoders and the blocklength under different performance criteria beyond the rate-distortion region. In this section, we present explicit definitions of fundamental limits of refined asymptotics.

Fix $(R_1^*,R_2^*)$ as an arbitrary rate pair on the boundary of the region $\calR_{\mathrm{inner}}(D_1,D_2)$. In second-order asymptotics under the JEP criterion, we characterize how a rate pair approaches $(R_1^*,R_2^*)$ with respect to the blocklength when a non-vanishing JEP is tolerated.
\begin{definition}\label{def:SOregion_JEP}
Given any $\varepsilon\in(0,1)$, a pair $(L_1,L_2)\in\bbR_+^2$ is said to be second-order $(R_1^*,R_2^*,D_1,D_2,\varepsilon|\dagger,\ddagger)$-achievable under JEP if there exists a sequence of $(n,M_1,M_2)$-codes in Definition \ref{def:srmismatch} such that
\begin{align}
\limsup_{n\to\infty}\frac{1}{\sqrt{n}}(\log M_1-nR_1^*)&\leq L_1, \label{def:SO_JEP_L1}\\
\limsup_{n\to\infty}\frac{1}{\sqrt{n}}(\log M_1M_2-nR_1^*-nR_2^*)&\leq L_2, \label{def:SO_JEP_L2}
\end{align}
 and
\begin{align}
\limsup_{n\to\infty}\rmP_{\dagger,\ddagger}^n(D_1,D_2|M_1,M_2)\leq\varepsilon.
\end{align}
The closure of the set of all second-order $(R_1^*,R_2^*,D_1,D_2,\varepsilon|\dagger,\ddagger)$-achievable rate pairs is denoted as $\calL(R_1^*,R_2^*,D_1,D_2,\varepsilon|\dagger,\ddagger)$.
\end{definition}

An equivalent form of \eqref{def:SO_JEP_L1} and \eqref{def:SO_JEP_L2} are as follows,
\begin{align}
\log M_1&\leq nR_1^*+\sqrt{n}L_1+o(\sqrt{n}),\\
\log M_1M_2&\leq n(R_1^*+R_2^*)+\sqrt{n}L_2+o(\sqrt{n}),
\end{align}
where $(L_1,L_2)$ are second-order rate pair.

Under SEP, we only define the second-order coding rate for the corner point of rate-distortion region, which corresponds to the minimal rates of both encoders. This is because the second-order coding rates are not well defined for other points on the boundary of the rate-distortion region for at least one encoder.

\begin{definition}\label{def:SOregion_SEP}
Given any $(\varepsilon_1,\varepsilon_2)\in(0,1)^2$, a pair $(L_1,L_2)\in\bbR_+^2$ is said to be second-order $(D_1,D_2,\varepsilon_1,\varepsilon_2|\dagger,\ddagger)$-achievable under SEP if there exists a sequence of $(n,M_1,M_2)$-codes such that
\begin{align}
\limsup_{n\to\infty}\frac{1}{\sqrt{n}}\left(\log M_1-\frac{n}{2}\log\frac{\sigma^2}{D_1}\right)&\leq L_1,\\
\limsup_{n\to\infty}\frac{1}{\sqrt{n}}\left(\log M_1M_2-\frac{n}{2}\log\frac{\sigma^2}{D_2}\right)&\leq L_2,
\end{align}
and
\begin{align}
\limsup_{n\to\infty}\rmP_{\dagger,\ddagger}^n(D_1|M_1,M_2)&\leq\varepsilon_1,\\
\limsup_{n\to\infty}\rmP_{\dagger,\ddagger}^n(D_2|M_1,M_2)&\leq\varepsilon_2.
\end{align}
The closure of the set of all second-order $(D_1,D_2,\varepsilon_1,\varepsilon_2|\dagger,\ddagger)$-achievable rate pairs is called the $(D_1,D_2,\varepsilon_1,\varepsilon_2|\dagger,\ddagger)$-achievable rate region and denoted as $\calL_{\rm{SEP}}(D_1,D_2,\varepsilon_1,\varepsilon_2|\dagger,\ddagger)$.
\end{definition}

In the moderate deviations regime, we are interested in a sequence of $(n,M_1,M_2)$-codes whose rates approach a boundary rate pair $(R_1^*,R_2^*)$ and whose excess-distortion probabilities vanish simultaneously. The definitions of the moderate deviations constants under JEP and SEP are given as follows.

Let $\theta_i$, $i\in[2]$ be two positive real numbers. Consider any sequence $\{\rho_n\}_{n\in\bbN}$ such that
\begin{align}
\rho_n\to0\;\mathrm{and}\;\sqrt{n}\rho_n\to\infty\mathrm{~as~}n\to\infty. \label{def:zeta_n}
\end{align}

\begin{definition}\label{def:Moderate_JEP}
A number $v\in\bbR_+$ is said to be a $(R_1^*,R_2^*,D_1,D_2,\theta_1,\theta_2|\dagger,\ddagger)$-achievable moderate deviations constant under JEP if there exists a sequence of $(n,M_1,M_2)$-codes such that
\begin{align}
\limsup_{n\to\infty}\frac{1}{n\rho_n}\left(\log M_1-nR_1^*\right)&\leq \theta_1,\\
\limsup_{n\to\infty}\frac{1}{n\rho_n}\big(\log M_1M_2-n(R_1^*+R_2^*)\big)&\leq \theta_1+\theta_2,
\end{align}
and
\begin{align}
\liminf_{n\to\infty}-\frac{\log\rmP_{\dagger,\ddagger}^n\left(D_1,D_2|M_1,M_2\right)} {n\rho_n^2}\geq v.
\end{align}
The supremum of all $(R_1^*,R_2^*,D_1,D_2,\theta_1,\theta_2|\dagger,\ddagger)$-achievable moderate deviations constants is denoted as \\ $v_{\dagger,\ddagger}^*(D_1,D_2|R_1^*,R_2^*,\theta_1,\theta_2)$.
\end{definition}

\begin{definition}
A pair $(v_1,v_2)\in\bbR_+^2$ is said to be a $(D_1,D_2,\theta_1,\theta_2|\dagger,\ddagger)$-achievable moderate deviations constant pair under SEP if there exists a sequence of $(n,M_1,M_2)$-codes such that
\begin{align}
\limsup_{n\to\infty}\frac{1}{n\rho_n}\left(\log M_1-\frac{n}{2}\log\frac{\sigma^2}{D_1}\right)&\leq \theta_1,\\
\limsup_{n\to\infty}\frac{1}{n\rho_n}\left(\log M_1M_2-\frac{n}{2}\log\frac{\sigma^2}{D_2}\right)&\leq \theta_1+\theta_2,
\end{align}
and
\begin{align}
\liminf_{n\to\infty}-\frac{\log\rmP_{\dagger,\ddagger}^n\left(D_1|M_1,M_2\right)} {n\rho_n^2}\geq v_1,\\
\liminf_{n\to\infty}-\frac{\log\rmP_{\dagger,\ddagger}^n\left(D_2|M_1,M_2\right)} {n\rho_n^2}\geq v_2.
\end{align}
The convex closure of the set of all $(D_1,D_2,\theta_1,\theta_2|\dagger,\ddagger)$-achievable moderate deviations constant pairs is called the\\ $(D_1,D_2,\theta_1,\theta_2|\dagger,\ddagger)$-achievable moderate deviations constant region and denoted as $\calV(D_1,D_2,\theta_1,\theta_2|\dagger,\ddagger)$.
\end{definition}

We next define the exponents of ensemble JEP and SEP when both encoders use i.i.d. Gaussian codebooks.
\begin{definition}
A number $E\in\bbR_+$ is said to be a $(D_1,D_2,R_1,R_2)$-achievable exponent under JEP if there exists a sequence of $(n,M_1,M_2)$-codes such that
\begin{align}
\limsup_{n\to\infty}\frac{1}{n}\log M_1&\leq R_1,\\
\limsup_{n\to\infty}\frac{1}{n}\log M_1M_2&\leq R_1+R_2,
\end{align}
and
\begin{align}
\liminf_{n\to\infty}-\frac{1}{n}\log \rmP_{\rm{iid}}^n(D_1,D_2|M_1,M_2)\geq E.
\end{align}
The supremum of all $(D_1,D_2,R_1,R_2)$ achievable exponents is denoted as $E^*(D_1,D_2|R_1,R_2)$.
\end{definition}

\begin{definition}
A pair $(E_1,E_2)\in\bbR_+^2$ is said to be $(D_1,D_2,R_1,R_2)$-achievable exponents under SEP if there exists a sequence of $(n,M_1,M_2)$-codes such that
\begin{align}
\limsup_{n\to\infty}\frac{1}{n}\log M_1&\leq R_1,\\
\limsup_{n\to\infty}\frac{1}{n}\log M_1M_2&\leq R_1+R_2,
\end{align}
and
\begin{align}
\liminf_{n\to\infty}-\frac{1}{n}\log \rmP_{\rm{iid}}^n(D_1|M_1,M_2)&\geq E_1,\\
\liminf_{n\to\infty}-\frac{1}{n}\log \rmP_{\rm{iid}}^n(D_2|M_1,M_2)&\geq E_2.
\end{align}
The convex closure of the set of all $(D_1,D_2,R_1,R_2)$-achievable exponents is denoted as $\calE(D_1,D_2,R_1,R_2)$.
\end{definition}

\section{Main Results}\label{sec:main_results}
In this section, we present our main results concerning the achievability analyses of second-order, moderate and large deviations asymptotics under both JEP and SEP. We consider all four combinations of spherical and i.i.d. Gaussian codebooks for both second-order and moderate deviations analyses and study large deviations only when both encoders use i.i.d. Gaussian codebooks.

\subsection{Second-Order Asymptotics}
Consider an arbitrary memoryless source that satisfies the moment constraint in \eqref{def:arbitrary_source} and two additional moment constraints:
\begin{align}
\bbE[X^4]&=:\zeta,\label{def:zeta}\\
\bbE[X^6]&<\infty.
\end{align}
Recall the definition of the ``mismatched'' dispersion in \cite{zhou2018refined}:
\begin{align}
  \rmV (\sigma^2,\zeta):=\frac{\zeta-\sigma ^4}{4\sigma^4} = \frac{\rm{Var}[X^2]}{4(\rm{E}[X^2]^2)}.
\end{align}

 We first present the result under JEP.
\begin{theorem}\label{theo:second_JEP}
Given any $\varepsilon\in(0,1)$, for any $\lambda\in(\frac{D_2}{D_1},1]$ and $(\dagger,\ddagger)\in\rm\{sp,iid\}^2$, there exists a sequence of $(n,M_1,M_2)$-codes using Gaussian codebooks such that
\begin{align}
\log M_1 &= \frac{n}{2}\log\frac{\sigma^2}{\lambda D_1} + \sqrt {n\rmV( \sigma ^2,\zeta)} Q^{ - 1}( \varepsilon  ) + O(\log n),\\
\log M_1 M_2&= \frac{n}{2}\log \frac{\sigma ^2}{D_2} + \sqrt {n\rmV( \sigma ^2,\zeta)} Q^{ - 1}( \varepsilon  ) + O(\log n),
\end{align}
and
\begin{align}
\lim_{n\to\infty}\rmP_{\dagger,\ddagger}^n(D_1,D_2|M_1,M_2)\leq \varepsilon.
\end{align}
\end{theorem}

The proof of Theorem \ref{theo:second_JEP} is provided in Section \ref{sec:proof_SOJEP}. A few remarks are in order.

Theorem \ref{theo:second_JEP} generalizes the achievability results in \cite[Theorem 1]{zhou2018refined} for the mismatched rate-distortion problem to the successive refinement setting. For any $\lambda\in(\frac{D_2}{D_1},1]$, when the target rate pair is $(R_1^*,R_2^*)=(\frac{1}{2}\log\frac{\sigma^2}{\lambda D_1},\frac{1}{2}\log\frac{\lambda D_1}{D_2})$, the pair $(L_1,L_2)=(\sqrt{\rmV(\sigma^2,\zeta)}Q^{-1}(\varepsilon), \sqrt{\rmV(\sigma^2,\zeta)}Q^{-1}(\varepsilon))$ is second-order $(R_1^*,R_2^*,D_1,D_2,\varepsilon|\dagger,\ddagger)$-achievable (cf. Definition \ref{def:SOregion_JEP}). This result generalizes \cite[Theorem 1]{zhou2018refined} by showing that all four combinations of spherical and i.i.d. Gaussian codebooks achieve the same second-order asymptotics.

At the first glance, Theorem \ref{theo:second_JEP} is not in the same form when presenting second-order asymptotics as \cite[Theorem 20]{zhou2016second} for a GMS in the matched case. However, with proper choice of the parameter $\lambda$ in Theorem \ref{theo:second_JEP}, we have the following inner bound to the second-order coding region $\calL(R_1^*,R_2^*,D_1,D_2,\varepsilon|\dagger,\ddagger)$.
\begin{itemize}
\item Case (i): Let $\lambda\in(\frac{D_2}{D_1},1)$, $R_1^*=\frac{1}{2}\log\frac{\sigma^2}{\lambda D_1}>\frac{1}{2}\log\frac{\sigma^2}{D_1}$,  $R_1^*+R_2^*=\frac{1}{2}\log\frac{\sigma^2}{D_2}$. It follows that
\begin{align}
\left\{(L_1,L_2): L_2\geq\sqrt{\rmV(\sigma^2,\zeta)}Q^{-1}(\varepsilon)\right\} \subseteq\calL(R_1^*,R_2^*,D_1,D_2,\varepsilon|\dagger,\ddagger).
\end{align}
\item Case (ii): Let $\lambda=1$, $R_1^*= \frac{1}{2}\log\frac{\sigma^2}{D_1}$ and $R_1^*+R_2^*>\frac{1}{2}\log\frac{\sigma^2}{D_2}$. It follows that
\begin{align}
\left\{(L_1,L_2):L_1\geq\sqrt{\rmV(\sigma^2,\zeta)}Q^{-1}(\varepsilon)\right\} \subseteq\calL(R_1^*,R_2^*,D_1,D_2,\varepsilon|\dagger,\ddagger).
\end{align}
\item Case (iii): Let $\lambda=1$, $R_1^*= \frac{1}{2}\log\frac{\sigma^2}{D_1}$ and $R_1^*+R_2^*=\frac{1}{2}\log\frac{\sigma^2}{D_2}$. It follows that
\begin{align}
\left\{(L_1,L_2):\min\{L_1,L_2\}\geq\sqrt{\rmV(\sigma^2,\zeta)}Q^{-1}(\varepsilon)\right\} \subseteq\calL(R_1^*,R_2^*,D_1,D_2,\varepsilon|\dagger,\ddagger).
\end{align}
\end{itemize}
Recall that the power of each codeword for encoder $f_1$ and encoder $f_2$ are $P_Y=\sigma^2-\lambda D_1$ and $P_Z=\lambda D_1-D_2$, respectively. If one uses a codebook with $\lambda=1$ as done in~\cite{wu2021}, we could only prove the results for cases (ii) and (iii).  When $\lambda<1$, the power of each codeword for encoder $f_1$ increases and the power of each codewords for encoder $f_2$ decreases. In the proof of Theorem 1, we show that $R_1^*=\frac{1}{2}\log\frac{\sigma^2}{\lambda D_1}>\frac{1}{2}\log\frac{\sigma^2}{D_1}$ and $R_2^*=\frac{1}{2}\log\frac{\lambda D_1}{D_2}<\frac{1}{2}\log\frac{D_1}{D_2}$ is asymptotically achievable, as a consequence of our refined second-order asymptotics. By introducing the parameter $\lambda\in(\frac{D_2}{D_1},1]$, we manage to cover all boundary points of the rate-distortion region, which allows us to derive achievable second-order asymptotics for all three cases of interest, especially for case (i) when the rate of encoder $f_1$ is larger than $\frac{1}{2}\log\frac{\sigma^2}{D_1}$.

When specialized to a GMS, the mismatched dispersion satisfies $\rmV(\sigma^2,\zeta)=\frac{1}{2}$. This implies that our proof of Theorem \ref{theo:second_JEP} is an alternative achievability proof for \cite[Theorem 20]{zhou2016second} with specific code design. Specifically, the achievability proof in \cite[Theorem 20]{zhou2016second} used a covering lemma in \cite{verger2005covering} that bounds the minimal number of required codewords to cover a ball \emph{without} providing the location of each codeword. In contrast, our proof of Theorem \ref{theo:second_JEP} specifies that the same second-order asymptotic performance can be achieved with an arbitrary combination of spherical and i.i.d. Gaussian codebooks using successive minimum Euclidean distance encoding.

It is a pity that we could not derive the ensemble converse result. The intuition is as follows. When either encoder $f_1$ or $f_2$ uses a spherical codebook, the codewords for the second pairs of encoder and decoder are not spherically symmetric, which makes it difficult to analyze the excess-distortion probability involving $D_2$. When both encoders used i.i.d. Gaussian codebooks, although the codeword distribution for the second pairs of encoder and decoder are i.i.d. Gaussian, with the parameter of $\lambda<1$, the target distortion for the first encoder-decoder pair is $D_1$ while the best achievable distortion from our codebook design is $\lambda D_1$. This additional mismatch between the target distortion level and the achievable distortion level for the first pair of encoder and decoder prevented us from obtaining ensemble tight results for $M_1$. Thus, the ensemble converse for mismatched successive refinement remains challenging and left as future work.

Note that when we ignore the first pair of encoder and decoder, the mismatched successive refinement problem considered in this paper reduces to a mismatched rate-distortion problem with superposition coding. It was recently shown in \cite[Lemma 4]{kanabar2022mismatched} by Kanabar and Scarlett that superposition coding could improve the compression rate for the ``mismatched" rate-distortion problem where the encoder and the decoder use different distortion measures. It would be of worthwhile to investigate whether superposition coding could improve the refined asymptotics of the mismatched rate-distortion problem in~\cite{zhou2018refined}.

We next present the result under SEP.
\begin{theorem}\label{theo:second_SEP}
Given any $\varepsilon_1,\varepsilon_2\in(0,1)^2$, for any $(\dagger,\ddagger)\in\rm\{sp,iid\}^2$,
\begin{align}
\bigg\{(L_1,L_2):L_1\geq\sqrt{\rmV(\sigma^2,\zeta)} \rmQ^{-1}\big(\min\{\varepsilon_1,\varepsilon_2\}\big), L_2\geq\sqrt{\rmV(\sigma^2,\zeta)} \rmQ^{-1}\big(\min\{\varepsilon_1,\varepsilon_2\}\big)\bigg\} \subseteq\calL_{\rm{SEP}}(D_1,D_2,\varepsilon_1,\varepsilon_2|\dagger,\ddagger).
\end{align}
\end{theorem}
The proof of Theorem \ref{theo:second_SEP} follows from the proof of Theorem \ref{theo:second_JEP} with $\lambda=1$ and $\varepsilon=\min\{\varepsilon_1,\varepsilon_2\}$.

A few remarks are as follows. Note that the minimum of the tolerable excess-distortion probabilities $(\varepsilon_1,\varepsilon_2)$ dominates the second-order achievable rate of both encoders. This is because under our mismatched coding scheme, only a non-excess-distortion event of decoder $\phi_1$ could guarantee a non-excess-distortion event of decoder $\phi_2$ with proper rate choices due to our layered coding scheme.

Theorem \ref{theo:second_SEP} generalizes the achievability part of the second-order asymptotic results in \cite{no2016} for the matched case to the mismatched scenario. When specialized to a GMS, we provide an alternative proof for \cite[Theorem 7]{no2016} by constructing a structured codebook with specific code design. It is important to note that the results in \cite[Theorem 7]{no2016} holds also when $\min\{\varepsilon_1,\varepsilon_2\}$ dominates the second-order coding rate since a similar layered coding scheme based on sphere covering was used\footnote{We remark that the claim in \cite[Theorem 7]{no2016} was flawed since $\varepsilon_1$ and $\varepsilon_2$ were used as the parameter of the $\rmQ^{-1}(\cdot)$ function in the second-order coding rates. This is because, in order not to incur an excess-distortion event at decoder $\phi_2$ for a sequence $x^n$, decoder $\phi_1$ should not incur an excess-distortion event since otherwise, the ``correct'' decoding of decoder $\phi_2$ is not guaranteed under the layered coding scheme.}.

Theorem \ref{theo:second_SEP} implies that when $\varepsilon_1=\varepsilon_2=\varepsilon$ for some $\varepsilon\in(0,1)$, any memoryless source satisfying \eqref{def:arbitrary_source} is strongly successive refinable~\cite[Definition 4]{no2016} under the quadratic distortion measure using our mismatched coding scheme. Following the proof of Theorem \ref{theo:second_JEP} and letting $\lambda=1$ and $\varepsilon=\min\{\varepsilon_1,\varepsilon_2\}$, we have that there exists a sequence of $(n,M_1,M_2)$ codes such that
\begin{align}
\log M_1&=\frac{n}{2}\log\frac{\sigma^2}{D_1}+\sqrt{n\rmV(\sigma^2,\zeta)} \rmQ ^{-1}(\varepsilon)+O(\log n),\\
\log M_1M_2 &=\frac{n}{2}\log\frac{\sigma^2}{D_2}+\sqrt{n\rmV(\sigma^2,\zeta)} \rmQ ^{-1}(\varepsilon)+O(\log n)\label{rate:encoder2},
\end{align}
and
\begin{align}
\lim_{n\to\infty}\rmP_{\dagger}^n(D_1|M_1)&\leq\varepsilon,\\
\lim_{n\to\infty}\rmP_{\ddagger}^n(D_2|M_1,M_2)&\leq\varepsilon\label{cons:dcoder2}.
\end{align}
Such a result is named strongly successive refinable since the sum rate in \eqref{rate:encoder2} is precisely the rate required to achieve \eqref{cons:dcoder2} even without the first layer of encoders and decoders, up to second-order asymptotics.

\subsection{Moderate Deviation Asymptotics}
Consider an arbitrary memoryless source with distribution $P_X$ satisfying \eqref{def:arbitrary_source} and \eqref{def:zeta} such that i) $\Lambda_{X^2}(\theta)$ is finite for some positive number $\theta$ and ii) the mismatched dispersion $\rmV(\sigma^2,\zeta)$ is finite.

\begin{theorem}
\label{theo:moderate_JEP}
For any rate pair $(R_1^*,R_2^*)$ on the boundary of $\calR_{\mathrm{inner}}$, any $(\dagger,\ddagger)\in\{\mathrm{sp},\mathrm{iid}\}^2$ and any real numbers $\lambda\in(\frac{D_2}{D_1},1]$, $(\theta_1,\theta_2)\in\mathbb{R}_+^2$,
\begin{align}
v_{\dagger,\ddagger}^*(D_1,D_2|R_1^*,R_2^*,\theta_1,\theta_2)\geq\frac{\theta_1^2}{2\rmV(\sigma^2,\zeta)}.
\end{align}
\end{theorem}
The proof of Theorem \ref{theo:moderate_JEP} is provided in Sections \ref{sec:proof_Mod_JEP_sp} and \ref{sec:proof_Mod_JEP_iid}. A few remarks are as follows.

We explain why  $v_{\dagger,\ddagger}^*(D_1,D_2|R_1^*,R_2^*,\theta_1,\theta_2)$ depends solely on $\theta_1$ and not $\theta_2$. We remark that $\theta_1$ is the speed at which $\log M_1$ deviates from $R_1^*$ and $\theta_2$ is the speed at which $\log M_2$ deviates from $R_2^*$ (cf. Definition 5). The dominant excess-distortion probability term related to $M_1$ scales as $\exp\left\{-\frac{n\theta_1^2\rho_n^2}{2\rmV(\sigma^2,\zeta)} +o(n\theta_1^2\rho_n^2)\right\}$ (cf. \eqref{Modspach:choose_M1}-\eqref{Modspach:solve_first_layer}) while the dominant excess-distortion probability term related to $M_2$ scales $\exp\{-n\theta_2\rho_n+o(\rho_n)\}$ (cf. \eqref{Modspach:choose_M2}-\eqref{Modspach:solve_second_layer}). Thus, $v_{\dagger,\ddagger}^*(D_1,D_2|R_1^*,R_2^*,\theta_1,\theta_2)$ is naturally dominant by $\theta_1$.

When specialized to a GMS, $\rmV(\sigma^2,\zeta)=\frac{1}{2}$. Our results in Theorem \ref{theo:moderate_JEP} recover the results of \cite[Theorem 21, Cases (ii) and (iii)]{zhou2016second}. The reason why we could not recover the optimal moderate deviations constant for a GMS in Case (i) where $R_1^*\in\left(\frac{1}{2}\log\frac{\sigma^2}{D_1},\frac{1}{2}\log\frac{\sigma^2}{D_2}\right)$ and $R_2^*=\frac{1}{2}\log\frac{\sigma^2}{D_2}$ is that under our mismatched coding scheme, the dominant error event is the excess-distortion event at decoder $\phi_1$ regardless of the rates of both encoders.

We next present the achievable moderate deviations constants under SEP.
\begin{theorem}\label{theo:moderate_SEP}
For any $(\dagger,\ddagger)\in\{\mathrm{sp},\mathrm{iid}\}^2$, $(\theta_1,\theta_2)\in\mathbb{R}_+^2$ and positive $\rmV(\sigma^2,\zeta)$, the achievable moderate deviations constants under SEP satisfy
\begin{align}
\left\{(v_1,v_2):v_1\geq\frac{\theta_1^2}{2\rmV(\sigma^2,\zeta)}, v_2\geq\frac{\theta_1^2}{2\rmV(\sigma^2,\zeta)}\right\} \subseteq\calV(D_1,D_2,\theta_1,\theta_2|\dagger,\ddagger).
\end{align}
\end{theorem}
The proof of Theorem \ref{theo:moderate_SEP} is provided in Section \ref{sec:proof_Mod_SEP}.

\subsection{Large Deviation Asymptotics}
\label{sec:largedeviations}
We only derive the exponent of ensemble JEP and SEP when i.i.d. Gaussian codebooks are used by both encoders since the large deviations performance of an i.i.d. Gaussian codebook is strictly better than the spherical codebook in the mismatched rate-distortion problem (cf. \cite[Lemma 4]{zhou2018refined}). To present our results, we need the following definitions. Given any $(s,w,P,D)\in\bbR_+^4$, define the following functions
\begin{align}
R_{\mathrm{iid}}(s,w,P,D)&:=\frac{1}{2}\log(1+2s)+\frac{sw}{(1+2s)P} -\frac{sD}{P}, \\
s^*(w,P,D)&:=\max\left\{0,\frac{P-2D+\sqrt{P^2+4wD}}{4D}\right\} , \label{def:s*}\\
R_{\rm{iid}}(w,P,D)&:=R_{\rm{iid}}(s^*(w,P,D),w,P,D).\label{def:Riid}
\end{align}

We remark that $R_{\mathrm{iid}}(w,P_Y,D_1)$ is the exponential decay rate of the non-excess-distortion probability with respect to
the distortion level $D_1$ for a source sequence $x^n$ with power $w=\frac{1}{n}||x^n||^2$ when its reproduction sequence is generated from $f_{\rm{iid}}(Y^n|\bzero^n,P_Y)$ (cf. \eqref{def:iidcodewords}), i.e.,
\begin{align}
\lim_{n\to\infty}-\frac{1}{n}\log\Pr_{f_{\rm{iid}}(Y^n|\bzero^n,P_Y)}\{d(x^n,Y^n)\leq D_1\}=R_{\mathrm{iid}}(w,P_Y,D_1).
\end{align}
Furthermore, $R_{\rm{iid}}(l,P_Z,D_2)$ is the the exponential decay rate of the non-excess-distortion probability with respect to
the distortion level $D_2$ for any source sequence $x^n$ whose quadratic distortion with respect to the output codeword $\hatx_1^n$ of encoder $f_1$ is $l:=d(x^n,\hatx_1^n)$ when its reproduction sequence is generated from  $f_{\rm{iid}}(Z^n|\hatx_1^n,P_Z)$, i.e.,
\begin{align}
\lim_{n\to\infty}-\frac{1}{n}\log\Pr_{f_{\rm{iid}}(Z^n|\hatx_1^n,P_Z)}\{d(x^n,Z^n)\leq D_2\}=R_{\mathrm{iid}}(l,P_Z,D_2).
\end{align}
Recall that $\Lambda_{X^2}^*(t)$ is the Fenchel-Legendre transform of the cumulant generating function of $X^2$, which is also known as the large deviation rate function~\cite[Chapter 2.3]{dembo2009large}.

We first present the achievable exponent of ensemble JEP under our mismatched coding scheme.
\begin{theorem}\label{theo:large_JEP}
Given any $(R_1,R_2)\in\bbR_+^2$, let $\lambda=\min\left\{\frac{D_2\exp\{2R_2\}}{D_1},1\right\}$ and let $\alpha^*(R_1,R_2)$ be the solution of $\alpha$ to $R_1= R_{\rm{iid}}(\alpha,P_Y,\lambda D_1)$. The joint excess-distortion exponent $E^*(D_1,D_2|R_1,R_2)$ satisfies:
\begin{align}
E^*(D_1,D_2|R_1,R_2)\geq\Lambda_{X^2}^*(\alpha^*(R_1,R_2)).
\end{align}
\end{theorem}
The proof of Theorem \ref{theo:large_JEP} is provided in Section \ref{sec:proof_large_JEP}, which generalizes the proof of \cite[Theorem 3]{zhou2018refined} for the rate-distortion problem to the successive refinement setting. We make the following remarks.

The joint excess-distortion exponent $\Lambda_{X^2}^*(\alpha^*(R_1,R_2))$ depends solely on $R_1$ if $R_2$ is large enough. Intuitively, this is because when the rate $R_2$ of encoder $f_2$ is large, the number of codewords used by $f_2$ is sufficient to $D_2$-cover the ball with center $\hatX_1^n$ and radius $D_1$. In this case, the excess-distortion probability of the second decoder $\phi_2$ with respect to the distortion level $D_2$ vanishes doubly exponentially fast (cf. \eqref{largeach:i_error_decompose_lambda} to \eqref{largeach:i_M_2_doubly}) and the joint excess-distortion event is dominated by the excess-distortion event of decoder $\phi_1$ (cf. \eqref{largeach:i_usePsiiid_(1-a)^M}).

The exponent $\Lambda_{X^2}^*(\alpha^*(R_1,R_2))$ is positive if the rate pair $(R_1,R_2)$ is strictly inside the rate region $\calR_{\mathrm{inner}}$ (cf. \eqref{ach:rdregion}). The reason is as follows. An equivalent form of $\calR_{\mathrm{inner}}$ is
\begin{align}
\calR_{\mathrm{inner}}=\bigcup_{\eta\in(\frac{D_2}{D_1},1]}\Big\{(R_1,R_2):~R_1\geq \frac{1}{2}\log\frac{\sigma^2}{\eta D_1}\mathrm{~and~}R_2\geq \frac{1}{2}\log\frac{\eta D_1}{D_2}\Big\}\label{equal:rinner}.
\end{align}
Note that \eqref{equal:rinner} follows from \eqref{ach:rdregion} since i) the constraint on rate $R_1$ in \eqref{equal:rinner} equals to the one in \eqref{ach:rdregion} by taking the union of $\eta\in(\frac{D_2}{D_1},1]$ and ii) adding the two inequalities together \eqref{equal:rinner} gives the the sum-rate constraint in \eqref{ach:rdregion}.
\begin{lemma}\label{theo:lemma_large_JEP_positive}
For any $\eta\in(\frac{D_2}{D_1},1]$, $\Lambda_{X^2}^*\big(\alpha^*(R_1,R_2)\big)>0$ if $R_1>\frac{1}{2}\log\frac{\sigma^2}{\eta D_1}$ and $R_2>\frac{1}{2}\log\frac{\eta D_1}{D_2}$.
\end{lemma}
\begin{proof}
Given any $\eta\in(\frac{D_2}{D_1},1]$, recall that $\lambda=\min\left\{\frac{D_2\exp\{2R_2\}}{D_1},1\right\}$, for $R_1>\frac{1}{2}\log\frac{\sigma^2}{\eta D_1}$ and $R_2>\frac{1}{2}\log\frac{\eta D_1}{D_2}$, it follows that
\begin{align}
\lambda&>\min\left\{\frac{D_2}{D_1}\frac{\eta D_1}{D_2},1\right\}\\
&=\eta.
\end{align}
Thus, it follows that $R_1>\frac{1}{2}\log\frac{\sigma^2}{\eta D_1}>\frac{1}{2}\log\frac{\sigma^2}{\lambda D_1}$.

Recall that $\alpha^*(R_1,R_2)$ is the solution of $\alpha$ to $R_1=R_{\rm{iid}}(\alpha,P_Y,\lambda D_1)$. It follows that
\begin{align}
R_{\rm{iid}}(\alpha^*(R_1,R_2),P_Y,\lambda D_1)&>\frac{1}{2}\log\frac{\sigma^2}{\lambda D_1}\\
&=R_{\rm{iid}}(\sigma^2,P_Y,\lambda D_1), \label{equal:use_Riid_sigma^2_lambdaD1}
\end{align}
where \eqref{equal:use_Riid_sigma^2_lambdaD1} follows from the definition of $R_{\rm{iid}}(\cdot)$ in \eqref{def:Riid} and Claim i) of Lemma \ref{theo:lemma_of_iid_large}.

Note that $R_{\rm{iid}}(w,P,D)$ increases in $w$ for $w\geq|D-P|^+$ (cf. Lemma \ref{theo:lemma_of_iid_large}, Claim iii)), we have that
\begin{align}
\alpha^*(R_1,R_2)>\sigma^2.
\end{align}
The proof is completed by noting that $\Lambda_{X^2}^*(\alpha)>0$ if $\alpha>\sigma^2$ (cf. Lemma \ref{theo:lemma_of_iid_large}, Claim iv)).
\end{proof}

When specialized to a GMS, $\Lambda_{X^2}^*(\alpha)=\mathop{\sup}_{\theta\in\bbR}\big\{\alpha \theta+\frac{1}{2}\log(1-2\sigma^2\theta)\big\} =\frac{1}{2}\left(\frac{\alpha}{\sigma^2}-\log\frac{\alpha}{\sigma^2}-1\right)$. To our best knowledge, this appears to be the first characterization of the excess-distortion exponent of successive refinement for continuous memoryless sources. To illustrate the exponent, we plot the $\Lambda_{X^2}^*(\alpha^*(R_1,R_2))$ for various values of $(R_1,R_2)$ in Fig. \ref{fig:Exponent_large} for a GMS. It is left as future work to check whether our achievable exponent is tight or not. To do so, one might generalize the result in~\cite{ihara2000error}, especially the converse part, to the successive refinement setting.

\begin{figure}[tb]
\centering
\includegraphics[width=0.5\columnwidth]{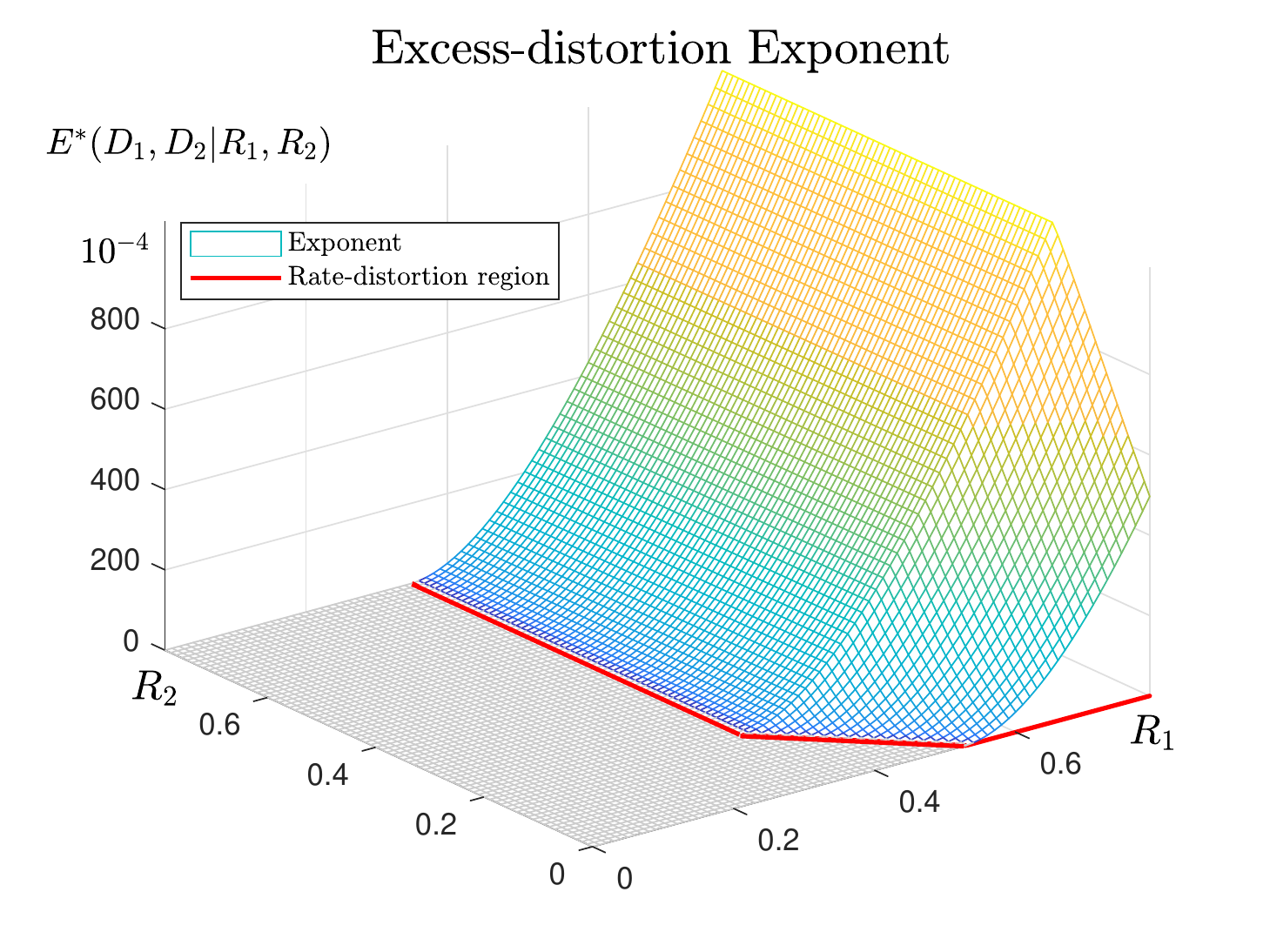}
\caption{Illustration of excess-distortion exponent $E^*(D_1,D_2|R_1,R_2)$ when $P_X=\calN(0,\sigma^2)$ with $\sigma^2=1$. The exponent maintains positive in rate-distortion region.}
\label{fig:Exponent_large}
\end{figure}

One might wonder why we use a coding scheme where the powers of codewords $P_Y=\sigma^2-\lambda D_1$ and $P_Z=\lambda D_1-D_2$ both have a parameter $\lambda$, in contrast to the mismatched rate-distortion case where one uses codewords of fixed powers $\sigma^2-D_1$ without any dummy variable~\cite{zhou2018refined}. As discussed in the second remark of Theorem \ref{theo:second_JEP}, such a code design allows us to derive results for all cases of interest in second-order asymptotics. For large deviations analysis, as we demonstrate above, this codebook design allows us to obtain positive exponent under JEP for all rate pairs strictly inside the rate region $\calR_{\mathrm{inner}}$, which is impossible for the codebook design with $\lambda=1$ as done in our previous ISIT paper~\cite{wu2022ISIT}. To discuss in detail, we first recall the main result in~\cite{wu2022ISIT}.
\begin{theorem}[{\cite[Theorem 1]{wu2022ISIT}}]\label{theo:large_iid_ISIT2022}
Consider the mismatched code in Definition \ref{def:srmismatch} where the codebook is designed with parameter $\lambda=1$. Given any $(R_1,R_2)\in\bbR_+^2$, let $\alpha_1$ be the solution to $R_1= R_{\rm{iid}}(\alpha_1,P_Y,D_1)$, $\gamma_2$ be the solution to $R_2=R_{\rm{iid}}(\gamma_2,P_Z,D_2)$ and $\alpha_2$ be the solution to $R_1= R_{\rm{iid}}(\alpha_2,P_Y,\gamma_2)$. The corresponding joint excess-distortion exponent $E_{\lambda=1}^*(D_1,D_2|R_1,R_2)$ satisfies
\begin{itemize}
\item Case i): If $R_1>\frac{1}{2}\log\frac{\sigma^2}{D_1}$ and $R_2>\frac{1}{2}\log\frac{D_1}{D_2}$,
\begin{align}
E_{\lambda=1}^*(D_1,D_2|R_1,R_2)\geq \Lambda_{X^2}^*(\alpha_1)>0.
\end{align}
\item Case ii): If $\frac{1}{2}\log\frac{D_1}{D_2}>R_2> R_{\rm{iid}}(|D_2-P_Z|^+,P_Z,D_2)$ and $R_1> R_{\rm{iid}}(\max\{\sigma^2,\gamma_2-P_Y\},P_Y,\gamma_2)$,
\begin{align}
E_{\lambda=1}^*(D_1,D_2|R_1,R_2)\geq \Lambda_{X^2}^*(\alpha_2)>0.
\end{align}
\item Case iii): Otherwise, $E_{\lambda=1}^*(D_1,D_2|R_1,R_2)=0$.
\end{itemize}
\end{theorem}
The condition for Case ii) is valid since $R_{\rm{iid}}(|D_2-P_Z|^+,P_Z,D_2)<\frac{1}{2}\log\frac{D_1}{D_2}$ and $R_{\rm{iid}}(\max\{\sigma^2,\gamma_2-P_Y\},P_Y,\gamma_2)\geq \frac{1}{2}\log\frac{\sigma^2}{D_1}$ (cf. Appendix \ref{subsec:justification_ISIT2022}). For ease of understanding, we illustrate all three cases of Theorem \ref{theo:large_iid_ISIT2022} in Fig. \ref{fig:rate_region_color} where $\calR_{\mathrm{inner}}$ is denoted as the region above the red line.

Note that if the rate pair $(R_1,R_2)\in\calR_{\mathrm{inner}}$ such that  $R_2<R_{\rm{iid}}(|D_2-P_Z|^+,P_Z,D_2)<\frac{1}{2}\log\frac{D_1}{D_2}$ and $R_1>R_{\rm{iid}}(\max\{\sigma^2,\gamma_2 -P_Y\},P_Y,\gamma_2)>\frac{1}{2}\log\frac{\sigma^2}{D_1}$, the exponent $E_{\lambda=1}^*(D_1,D_2|R_1,R_2)=0$. Intuitively, this is because using the $\lambda=1$ code, regardless of the rate $R_1>\frac{1}{2}\log\frac{\sigma^2}{D_1}$ of encoder $f_1$, the distortion between a source sequence $x^n$ and the output $\hatx_1^n$ of $f_1$ is roughly $D_1$. In this case, a rate $R_2<\frac{1}{2}\log\frac{D_1}{D_2}$ of for encoder $f_2$ is not sufficient to $D_2$-cover the distortion ball of the source sequence with radius $D_1$. This problem can be resolved by introducing a rate related parameter $\lambda\in(\frac{D_2}{D_1},1]$ into the codebook design as we do in the present paper. This way, the average power of codewords for encoder $f_1$ is $P_Y=\sigma^2-\lambda D_1$, hence the achievable distortion level of encoder $f_1$ reduces to $\lambda D_1$ for $R_1=\frac{1}{2}\log\frac{\sigma^2}{\lambda D_1}>\frac{1}{2}\log\frac{\sigma^2}{D_1}$. Furthermore, the average power of codewords for encoder $f_2$ is now $P_Z=\lambda D_1-D_2$ so that encoder $f_2$ is able to $D_2$-cover the distortion ball of a source sequence with radius $\lambda D_1$ with rate $R_2>\frac{1}{2}\log\frac{\lambda D_1}{D_2}$, which is smaller than $\frac{1}{2}\log\frac{D_1}{D_2}$. As shown in Fig. \ref{fig:rate_region_gap}, our present rate-dependent coding scheme grantees a positive exponent $E^*(D_1,D_2|R_1,R_2)$ for all $(R_1,R_2)\in\calR_{\mathrm{inner}}$.

\begin{figure}[tb]
\centering
\includegraphics[width=.5\columnwidth]{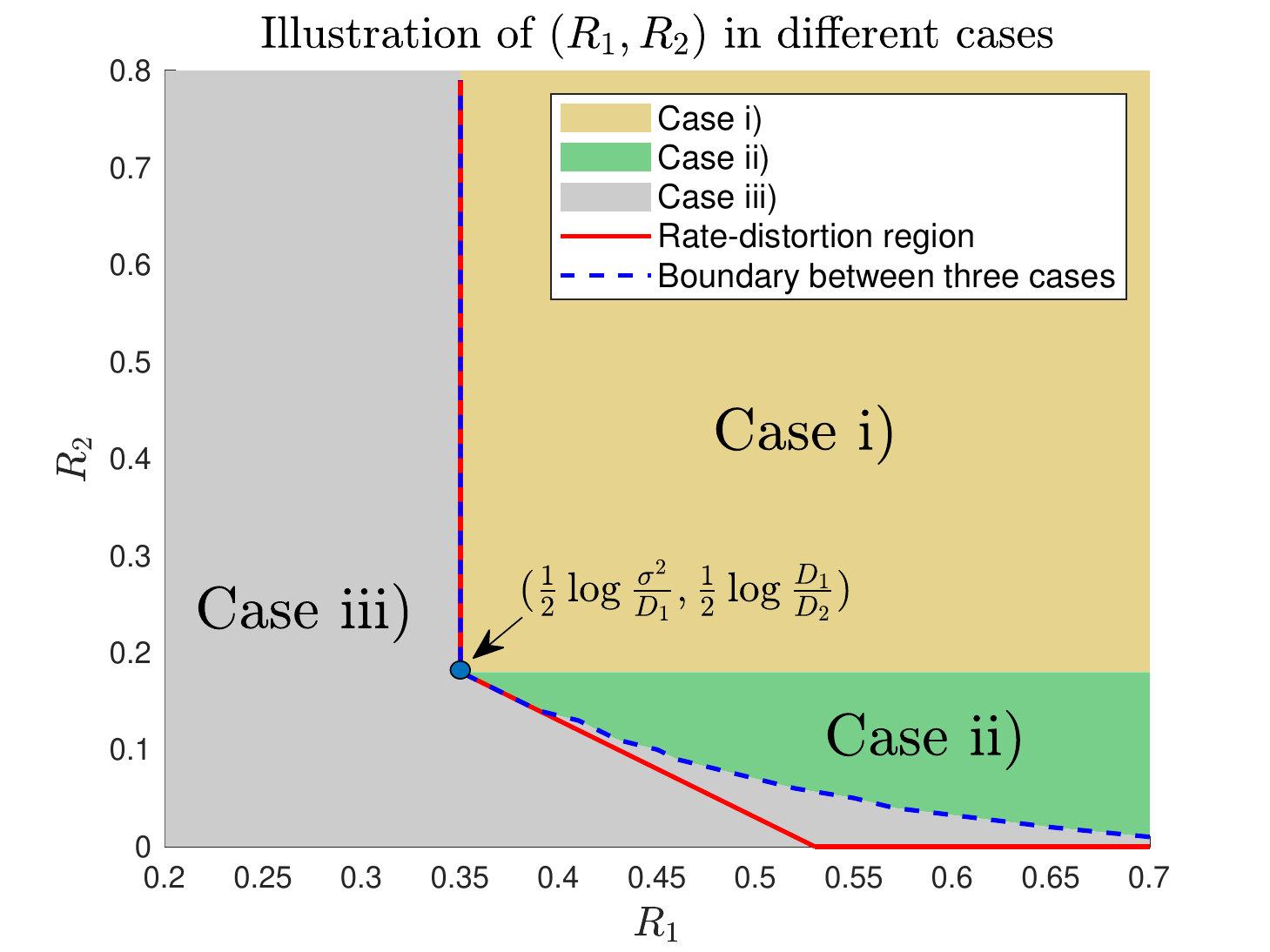}
\caption{Illustration of different cases in Theorem \ref{theo:large_iid_ISIT2022}. The exponent $E_{\lambda=1}^*(D_1,D_2|R_1,R_2)$ is positive in case i) and case ii) but equals to 0 in case iii). }
\label{fig:rate_region_color}
\end{figure}
\begin{figure}[tb]
\centering
\includegraphics[width=.5\columnwidth]{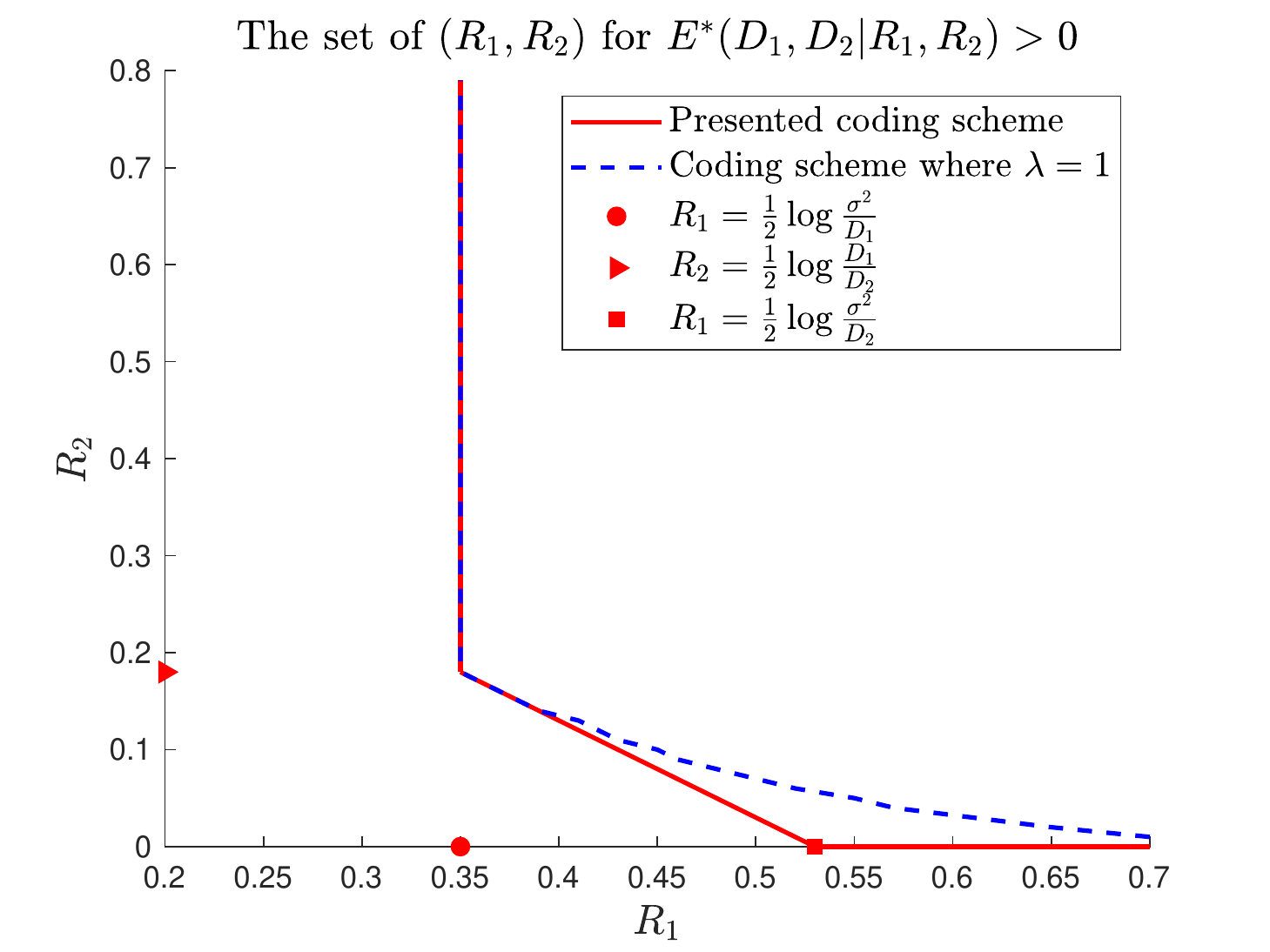}
\caption{Comparison of the set of rate pairs $(R_1,R_2)$ that ensure a positive JEP exponent $E^*(D_1,D_2|R_1,R_2)$ for a GMS $P_X=\calN(0,1)$ between our present coding scheme in Theorem \ref{theo:large_JEP} and coding scheme in Theorem \ref{theo:large_iid_ISIT2022} where $\lambda=1$. Note that such rate pairs are above the boundary denoted by red and blue dotted lines for the two coding schemes, respectively.}
\label{fig:rate_region_gap}
\end{figure}

We next present achievable exponents of ensemble SEP. Recall that $\lambda=\min\left\{\frac{D_2\exp\{2R_2\}}{D_1},1\right\}$.
\begin{theorem}\label{theo:large_SEP}
Given any $(R_1,R_2)\in\bbR_+^2$, let $\alpha_1^*(R_1)$ be the solution of $\alpha$ to $R_1=R_{\rm{iid}}(\alpha,P_Y,D_1)$. The following results hold.
\begin{itemize}
\item If $R_2\leq\frac{1}{2}\log\frac{D_1}{D_2}$, let $\alpha_2^*(R_1)$ be the solution of $\alpha$ to $R_1= R_{\rm{iid}}(\alpha,P_Y,\lambda D_1)$. The exponent region satisfies
\begin{align}
\{(E_1,E_2):E_1\geq\Lambda_{X^2}^*(\alpha_1^*(R_1)),E_2\geq\Lambda_{X^2}^*(\alpha_2^*(R_1))\} \subseteq\calE(D_1,D_2,R_1,R_2).
\end{align}
\item If $R_2>\frac{1}{2}\log\frac{D_1}{D_2}$, note that $\lambda=1$, let $\gamma^*(R_2)$ be the solution of $\gamma$ to $R_2=R_{\rm{iid}}(\gamma,P_Z,D_2)$ and let $\alpha_2^*(R_1,R_2)$ be the solution of $\alpha$ to $R_1= R_{\rm{iid}}(\alpha,P_Y,\gamma^*(R_2))$. The exponent region satisfies
\begin{align}
\{(E_1,E_2):E_1\geq\Lambda_{X^2}^*(\alpha_1^*(R_1)), E_2\geq\Lambda_{X^2}^*(\alpha_2^*(R_1,R_2))\}\subseteq\calE(D_1,D_2,R_1,R_2).
\end{align}
\end{itemize}
\end{theorem}
The proof of Theorem \ref{theo:large_SEP} is provided in Section \ref{sec:proof_large_SEP}. For ease of notation, let $E^*(D_1|R_1,R_2):=\Lambda_{X^2}^*(\alpha_1^*(R_1))$ and let $E^*(D_2|R_1,R_2)$ be
\begin{align}
E^*(D_2|R_1,R_2)
:=\left\{
\begin{array}{ll}
\Lambda_{X^2}^*(\alpha_2^*(R_1))&\mathrm{if}\;\;R_2\leq\frac{1}{2}\log\frac{D_1}{D_2},\\
\Lambda_{X^2}^*(\alpha_2^*(R_1,R_2))&\mathrm{otherwise}.
\end{array}
\right.
\end{align}
The following lemma states the condition under which $E^*(D_1|R_1,R_2)$ and $E^*(D_2|R_1,R_2)$ are positive.
\begin{lemma}
The following claims hold.
\begin{enumerate}[i)]
\item The exponent $E^*(D_1|R_1,R_2)>0$ if $R_1>\frac{1}{2}\log\frac{\sigma^2}{D_1}$.
\item The exponent  $E^*(D_2|R_1,R_2)>0$ if the rate pair $(R_1,R_2)$ is strictly inside the rate region $\calR_{\mathrm{inner}}$.
\item Both $E^*(D_1|R_1,R_2)$ and $E^*(D_2|R_1,R_2)$ are positive if $(R_1,R_2)$ is strictly inside the rate region $\calR_{\mathrm{inner}}$.
\end{enumerate}
\end{lemma}
\begin{proof}
Claim i) is established as follows. Recall that $E^*(D_1|R_1,R_2)=\Lambda_{X^2}^*(\alpha_1^*(R_1))$ and $\alpha_1^*(R_1)$ be the solution of $\alpha$ to $R_1=R_{\rm{iid}}(\alpha,P_Y,D_1)$. For $R_1>\frac{1}{2}\log\frac{\sigma^2}{D_1}$, it follows that
\begin{align}
R_{\rm{iid}}(\alpha_1^*(R_1),P_Y,D_1)>\frac{1}{2}\log\frac{\sigma^2}{D_1}.
\end{align}
Note that $R_{\rm{iid}}(\sigma^2,P_Y,D_1)=\frac{1}{2}\log\frac{\sigma^2}{D_1}$ and $R_{\mathrm{iid}}(w,P,D)$ increases in $w$ for $w\geq|D-P|^+$ (cf. Lemma \ref{theo:lemma_of_iid_large}, Claim i) and iii)), we have that
\begin{align}
\alpha_1^*(R_1)>\sigma^2.
\end{align}
The proof of Claim i) is completed by noting that  $\Lambda^*_{X^2}(a)>0$ when $a>\sigma^2$ (cf. Lemma \ref{theo:lemma_of_iid_large}, Claim iv)).

Claim ii) is proved as follows. When $R_2\leq\frac{1}{2}\log\frac{D_1}{D_2}$, since $E^*(D_2|R_1,R_2)=E^*(D_1,D_2|R_1,R_2)$, it follows from the positivity analysis of $E^*(D_1,D_2|R_1,R_2)$ (cf. Lemma \ref{theo:lemma_large_JEP_positive}) that $E^*(D_2|R_1,R_2)$ is positive if $(R_1,R_2)$ is strictly inside the rate region $\calR_{\mathrm{inner}}$.

When $R_2>\frac{1}{2}\log\frac{D_1}{D_2}$, recall that $\lambda=1$, the rate pair $(R_1,R_2)$ is strictly inside the rate region $\calR_{\mathrm{inner}}$ if and only if $R_1>\frac{1}{2}\log\frac{\sigma^2}{D_1}$. Recall that $\gamma^*(R_2)$ is the solution of $\gamma$ to $R_2=R_{\rm{iid}}(\gamma,P_Z,D_2)$, invoking that $R_{\rm{iid}}(l,P,D)$ increases in $l$ when $l\geq|D-P|^+$ and $R_{\rm{iid}}(D_1,P_Z,D_2)=\frac{1}{2}\log\frac{D_1}{D_2}$ (cf. Lemma \ref{theo:lemma_of_iid_large}, Claim i) and iii)), it follows that
\begin{align}
\gamma^*(R_2)>D_1.
\end{align}
Recall that $\alpha_2^*(R_1,R_2)$ is the solution of $\alpha$ to $R_1= R_{\rm{iid}}(\alpha,P_Y,\gamma^*(R_2))$; it follows that
\begin{align}
R_{\rm{iid}}(\alpha_2^*(R_1,R_2),P_Y,\gamma^*(R_2))&=R_1\\
&>\frac{1}{2}\log\frac{\sigma^2}{D_1}\\
&=R_{\rm{iid}}(\sigma^2,P_Y,D_1)\label{equal:use_Riid_sigma^2_D1},
\end{align}
where \eqref{equal:use_Riid_sigma^2_D1} follows from Lemma \ref{theo:lemma_of_iid_large}, Claim i). Invoking that $R_{\rm{iid}}(l,P,D)$ increases in $l$ when $l\geq|D-P|^+$, decreases in $D$ when $D\leq l+P$ (cf. Lemma \ref{theo:lemma_of_iid_large}, Claim iii)) and $\gamma^*(R_2)>D_1$, it follows that
\begin{align}
\alpha_2^*(R_1,R_2)>\sigma^2.
\end{align}
Therefore, $E^*(D_2|R_1,R_2)>\Lambda_{X^2}^*(\alpha_2^2(R_1,R_2))>\Lambda_{X^2}^*(\sigma^2)>0$.

Claim iii) is established as follows. Note that $\frac{1}{2}\log\frac{\sigma^2}{\lambda D_1}\geq\frac{1}{2}\log\frac{\sigma^2}{D_1}$ for $\lambda\in(\frac{D_2}{D_1},1]$, we have that $E^*(D_2|R_1,R_2)>0$ implies $E^*(D_1|R_1,R_2)>0$. Combining the above arguments, we find that both $E^*(D_1|R_1,R_2)$ and $E^*(D_2|R_1,R_2)$ are positive if $(R_1,R_2)$ is strictly inside the rate region $\calR_{\mathrm{inner}}$.
\end{proof}

When specialized to a GMS, $\Lambda_{X^2}^*(\alpha_2)=\frac{1}{2}\big(\frac{\alpha_2}{\sigma^2}- \log\frac{\alpha_2}{\sigma^2}-1\big)$. This result appears the first attempt to derive a large deviations result for GMS under SEP in the successive refinement problem. It is interesting future work to check whether our exponents are tight or not. To illustrate the exponent, we plot $E^*(D_2|R_1,R_2)$ in Fig. \ref{fig:Exponent_large_SEP} for a GMS.

\begin{figure}[tb]
\centering
\includegraphics[width=.5\columnwidth]{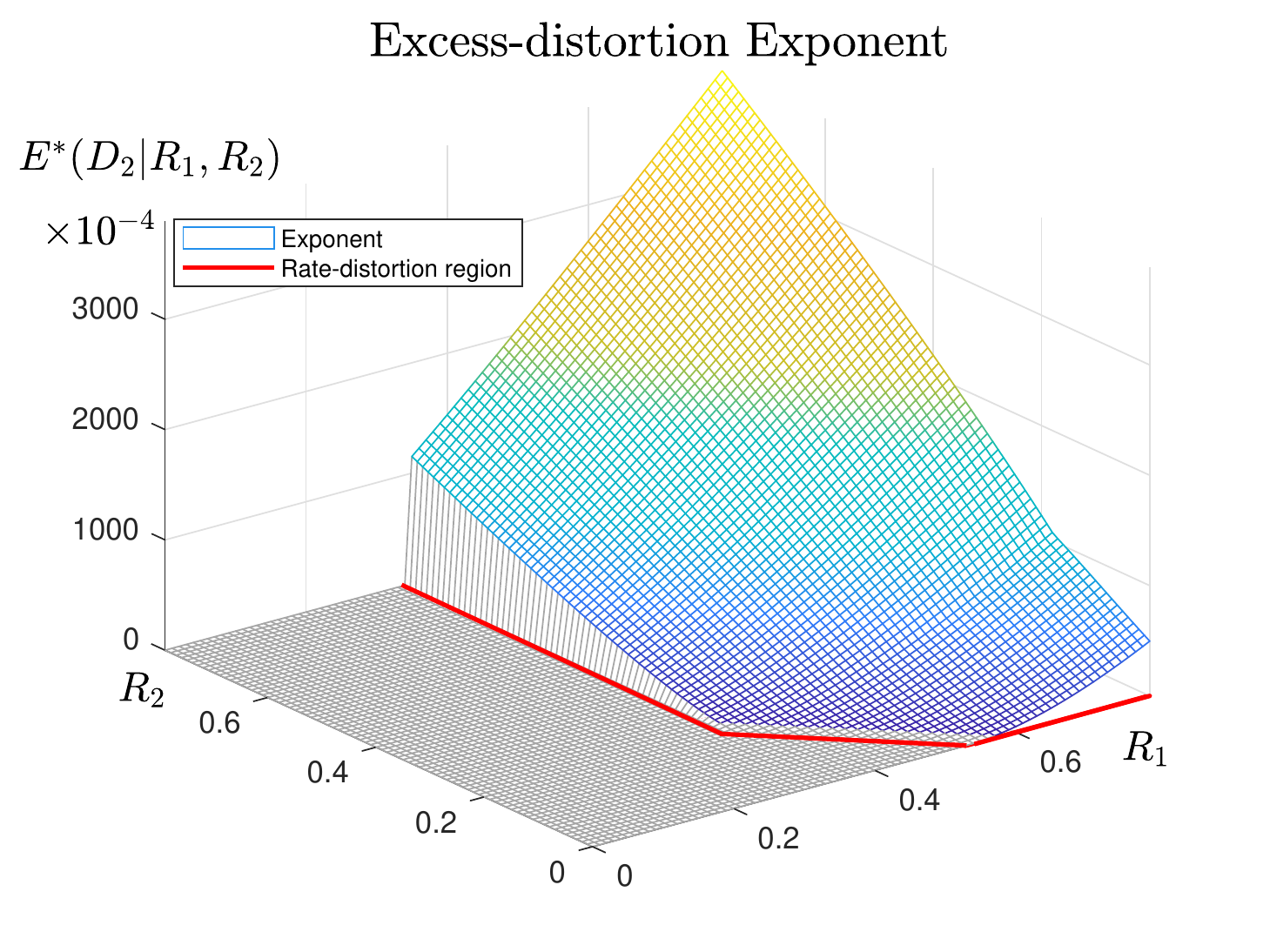}
\caption{Illustration of excess-distortion exponent $E^*(D_2|R_1,R_2)$ when $P_X=\calN(0,1)$.}
\label{fig:Exponent_large_SEP}
\end{figure}

\section{Proof of Second-Order Asymptotics under JEP (Theorem \ref{theo:second_JEP})}\label{sec:proof_SOJEP}

\subsection{Preliminaries}
Recall that $P_Y=\sigma^2-\lambda D_1$ and $P_Z=\lambda D_1-D_2$. For any $\varepsilon \in (0,1)$, let
\begin{align}
\rmV &:=\mathrm{Var} [X^2]=\zeta - \sigma ^4,\label{def:V}\\
a_n &:=\sqrt{\rmV\frac{{\log n}}{n}}\label{def:a_n},\\
b_n &:=\sqrt{\frac{\rmV}{n}}\rmQ^{-1}(\varepsilon).\label{def:b_n}
\end{align}

In subsequent analysis, for simplicity, let the random variable $W$ denote the  normalized $\ell_2$ norm of source sequence, i.e., $\frac{\|X^n\|^2}{n}$. We use $f_W$ to denote the distribution of random variable $W$ and use $w$ as a particular realization of $W$. Similarly, we define the random variable $L$ as the quadratic distortion of the first layer, i.e., $L=d(X^n,\hatX_1^n)$. Furthermore, we use $f_L$ to denote the induced distribution of the random variable $L$ and use $l$ as a particular realization of $L$.

For any $(x^n,y^n)$ and $(\dagger,\ddagger)\in\{\rm{sp},\rm{iid}\}^2$, define following non-excess-distortion probabilities
\begin{align}
\Psi_{\dagger}(n,w,D)&:= \Pr_{f_{\dagger}(\cdot|\bzero^n,P_Y)} \{ d(x^n,Y^n) \le D\},\label{def:Psi_sp}\\
\Phi_{\ddagger}(n,x^n,y^n)&:= \Pr_{f_{\ddagger}(\cdot|y^n,P_Z)}\{d(x^n,Z^n)\leq D_2\},\label{def:Phi_sp}
\end{align}
where $f_{\cdot}(\cdot)$ is the distribution of codewords (cf. Definition \ref{def:srmismatch}). Define two constants
\begin{align}
\beta_1&:=\sqrt{P_Z}-\sqrt{D_2}, \label{def:beta1}\\
\beta_2 &:=\sqrt{P_Z}+\sqrt{D_2} \label{def:beta2}.
\end{align}

It follows that
\begin{align}
\beta_1^2\leq|P_Z-D_2|\label{SOspach:beta1_leq}.
\end{align}
We justify \eqref{SOspach:beta1_leq} as follows. Note that $\beta_1^2=\lambda D_1-2\sqrt{P_Z}\sqrt{D_2}$. If $\lambda D_1\geq 2D_2$, then $\sqrt{P_Z}\geq \sqrt{D_2}$ and $|P_Z-D_2|=\lambda D_1-2D_2$, and thus
\begin{align}
\beta_1^2
&=\lambda D_1-2\sqrt{P_Z}\sqrt{D_2}\\
&\leq \lambda D_1-2D_2\\
&=|P_Z-D_2|.
\end{align}
If $\lambda D_1<2D_2$, then $\sqrt{P_Z}<\sqrt{D_2}$ and $|P_Z-D_2|=2D_2-\lambda D_1$, and thus
\begin{align}
\beta_1^2-|P_Z-D_2|
&=\lambda D_1-2\sqrt{P_Z}\sqrt{D_2}-2D_2+\lambda D_1\\
&=2(P_Z-\sqrt{P_Z}\sqrt{D_2})\\
&=2\sqrt{P_Z}(\sqrt{P_Z}-\sqrt{D_2})\\
&<0.
\end{align}

Note that \eqref{def:Psi_sp} is valid since $\Pr_{f_{\dagger}(\cdot|\bzero^n,P_Y)} \{d(x^n,Y^n)\le D\}$ depends on $x^n$ only through its normalized $\ell_2$ norm $w=\frac{\|x^n\|^2}{n}$. Similarly to~\cite[Theorem 37]{kostina2012}, $\Phi_{\mathrm{sp}}(n,x^n,y^n)$ depends on $(x^n,y^n)$ only through their quadratic distortion $l:=d(x^n,y^n)$ and we obtain the following lemma.

\begin{lemma}\label{theo:lemma_SO_sp_ach}
Recall that $|x|^+=\max\{x,0\}$. For any $(x^n,y^n)$ such that $\sqrt{l}\in[|\beta_1|^+,\beta_2]$
\begin{align}
\Phi_{\mathrm{sp}}(n,x^n,y^n)
& \geq \frac{\Gamma(\frac{n+2}{2})}{\sqrt\pi n\Gamma(\frac{n+1}{2} )}\left(1-\frac{(l+P_Z-D_2)^2}{4lP_Z}\right)^\frac{n-1}{2} \nonumber \\
& =:\underline{h}(n,l),\label{def:h}
\end{align}
and otherwise $\Phi_{\mathrm{sp}}(n,x^n,y^n)=0$.
\end{lemma}
The proof of Lemma \ref{theo:lemma_SO_sp_ach} is available in Appendix \ref{subsec:proof_lemma_SO_sp_ach}.

We also need the following function
\begin{align}
\kappa(s,w,P)&:=\frac{(P(1+2s)+2w)^2}{P(1+2s)^3}\label{def:kappa}.
\end{align}
Recall the definition of $R_{\mathrm{iid}}(w,P,D)$, $s^*(w,P,D)$ and $\kappa(s,w,P)$ in \eqref{def:s*}, \eqref{def:Riid} and \eqref{def:kappa}, respectively. Similar to Zhou et. al \cite[cf. (89)]{zhou2018refined}, using the strong large deviations theorem in~\cite[Theorem 3.7.4]{dembo2009large}, noting that $\Phi_{\mathrm{iid}}(n,x^n,y^n)$ depends on $(x^n,y^n)$ only through their quadratic distortion $l=d(x^n,y^n)$, we obtain
\begin{align}
\Phi_{\mathrm{iid}}(n,l)\sim\frac{\exp\{-nR_{\mathrm{iid}}(l,P_Z,D_2)\}} {s^*(l,P_Z,D_2)\sqrt{\kappa(s^*(l,P_Z,D_2),l,P_Z,D_2)}}. \label{Soiidach:result_strong_large}
\end{align}

Finally, note that for any $(\dagger,\ddagger)\in\{\rm{sp},\rm{iid}\}^2$, the ensemble JEP satisfies
\begin{align}
\rmP_{\dagger,\ddagger}^n(D_1,D_2|M_1,M_2)
&\leq \Pr\{d(X^n,\hatX_1^n)>\lambda D_1\;\mathrm{or}\;d(X^n,\hatX_2^n)>D_2 \}\\
&=\Pr\{d(X^n,\hat{X}_1^n)>\lambda D_1\}+\Pr\{d(X^n,\hat{X}_1^n)\leq \lambda D_1, \; d(X^n,\hat{X}_2^n)>D_2\}. \label{SOspach:error_decompose}
\end{align}
We next further upper bound \eqref{SOspach:error_decompose}.

\subsection{When Both Encoders Use Spherical Codebooks}\label{sec:proof_SOJEP_sp}
We first consider the case where both encoders use spherical codebooks where $(\dagger,\ddagger)=\mathrm\{\rm{sp},\rm{sp}\}$. Note that the first term in \eqref{SOspach:error_decompose} is similar to the mismatched rate-distortion problem studied in~\cite[Section IV. B]{zhou2018refined}, except that i) the distortion level is replaced from $D$ to $\lambda D_1$, and ii) the power of codebook is changed from $P_Y=\sigma^2-D_1$ to $P_Y=\sigma^2-\lambda D_1$. Choose $M_1$ such that
\begin{align}
\log M_1=\frac{n}{2}\log\frac{\sigma^2}{\lambda D_1} + \sqrt {n\rm{V}(\sigma^2,\zeta)}\rmQ^{-1}(\varepsilon)+O(\log n). \label{SOspach:chooseM_1}
\end{align}
Similar to~\cite[cf. (43)-(58)]{zhou2018refined}, we have that
\begin{align}
\Pr\{d(X^n,\hat{X}_1^n)>\lambda D_1\}\leq\varepsilon+O\left(\frac{1}{\sqrt{n}}\right).\label{SOspach:solve_first_layer}
\end{align}

Thus, we focus on the analysis of the second term in \eqref{SOspach:error_decompose}. For ease of notation, we use $\bY$ to denote the random codebooks $(Y^n(1),\ldots,Y^n(M_1))$
of the first encoder and use $\by$ for a particular realization. Recall that given the source sequence $x^n$ and a particular realization of the codebook $\by$, the compressed index of the encoder $f_1$ is $f_1(x^n)$ and the reproduced source sequence of the first encoder is
\begin{align}
\hatx_1^n=y^n(f_1(x^n))=\min_{i\in[M_1]}d(x^n,y^n(i)).
\end{align}
It follows that
\begin{align}
\nn&\Pr\{ d(X^n,\hat{X}_1^n)\leq \lambda D_1, \; d(X^n,\hat{X}_2^n)>D_2 \}\\*
&=\int_{\substack{x^n,\by:\\d(x^n,\hatx_1^n)\leq \lambda D_1}}P_X^n(x^n)\prod_{i\in[M_1]}f_{\rm{sp}}(y^n(i)|\bzero^n,P_Y)
\times \Pr\{\forall~j\in[M_2]:~d(x^n,Z^n(f_1(x^n),j))>D_2\}\rmd x^n\rmd \by \label{SOspach:decompose_second_layer}\\
&=\int_{\substack{x^n,\by:\\d(x^n,\hatx_1^n)\leq \lambda D_1}}P_X^n(x^n)\prod_{i\in[M_1]}f_{\rm{sp}}(y^n(i)|\bzero^n,P_Y)\times (1-\Phi_{\mathrm{sp}}(n,x^n,\hatx_1^n))^{M_2}\rmd x^n\rmd \by\label{SOspach:useindependent}\\
\nn&=\int_{\substack{x^n,\by}}P_X^n(x^n)\prod_{i\in[M_1]} f_{\rm{sp}}(y^n(i)|\bzero^n,P_Y)\times 1\{d(x^n,\hatx_1^n)<(|\beta_1|^+)^2\} \times(1-\Phi_{\mathrm{sp}}(n,x^n,\hatx_1^n))^{M_2}\rmd x^n\rmd\by\\* &\quad+\int_{\substack{x^n,\by}}P_X^n(x^n)\prod_{i\in[M_1]} f_{\rm{sp}}(y^n(i)|\bzero^n,P_Y)\times 1\{d(x^n,\hatx_1^n)\in[(|\beta_1|^+)^2,\lambda D_1]\}\times(1-\Phi_{\mathrm{sp}}(n,x^n,\hatx_1^n))^{M_2}\rmd x^n\rmd\by \label{Sospach:involve_indicator}\\
\nn&=\Pr\{d(X^n,\hatX_1^n)<(|\beta_1|^+)^2\} +\int_{\substack{x^n,\by}}P_X^n(x^n)\prod_{i\in[M_1]} f_{\rm{sp}}(y^n(i)|\bzero^n,P_Y)\times 1\{d(x^n,\hatx_1^n)\in[(|\beta_1|^+)^2,\lambda D_1]\}\\*
&\qquad\times(1-\Phi_{\mathrm{sp}}(n,x^n,\hatx_1^n))^{M_2}\rmd x^n\rmd\by\label{SOspach:useproperty},
\end{align}
where \eqref{SOspach:useindependent} follows since given $x^n$, $\by$ and thus $\hatx_1^n$, the codewords $(Z^n(\hatx_1^n,1),\ldots,Z^n(\hatx_1^n,M_2))$ are independent and generated by the same distribution $f_{\rm{sp}}(\cdot|\hatx_1^n,\lambda D_1-D_2)$, \eqref{Sospach:involve_indicator} divides the whole section into two parts and the second term is valid since $(|\beta_1|^+)^2<\lambda D_1$ and \eqref{SOspach:useproperty} follows from the results in Lemma \ref{theo:lemma_SO_sp_ach}. 

Recall that $l=d(x^n,y^n)$ is the quadratic distortion between $x^n$ and $y^n$. The second term in \eqref{SOspach:useproperty} can be further upper bounded by
\begin{align}
&\int_{(|\beta_1|^+)^2}^{\lambda D_1}f_L(l)(1-\underline{h}(n,l))^{M_2}\rmd l \label{SOspach:useh}\\
&\leq \int_{(|\beta_1|^+)^2}^{|P_Z-D_2|}f_L(l)\rmd l+\int_{|P_Z-D_2|}^{\lambda D_1}f_L(l)(1-\underline{h}(n,l))^{M_2}\rmd l \label{SOspach:removebeta}\\
&\leq \int_{(|\beta_1|^+)^2}^{|P_Z-D_2|}f_L(l)\rmd l+\int_{|P_Z-D_2|}^{\lambda D_1}f_L(l)\exp\{-M_2\underline{h}(n,l)\}\rmd l
\label{SOspach:usingexp{-Ma}}\\
&\leq\int_{(|\beta_1|^+)^2}^{|P_Z-D_2|}f_L(l)\rmd l+\exp\{-M_2\underline{h}(n,\lambda D_1)\}, \label{SOspach:removeint_byD_1}
\end{align}
where \eqref{SOspach:useh} follows from the result in \eqref{def:h}, \eqref{SOspach:removebeta} follows since $\beta_1^2\leq|P_Z-D_2|$ (cf. \eqref{SOspach:beta1_leq}) and $\underline{h}(n,l)\in[0,1]$, \eqref{SOspach:usingexp{-Ma}} follows from that $(1-a)^{M} \leq \exp\{-Ma \}$ for any $a\in [0,1)$ and \eqref{SOspach:removeint_byD_1} follows since $\underline{h}(n,l)$ is a decreasing function of $l$ if $l\geq|P_Z-D_2|$.

Combining \eqref{SOspach:useproperty} and \eqref{SOspach:removeint_byD_1}, we have
\begin{align}
\Pr\{d(X^n,\hat{X}_1^n)\leq\lambda D_1,\;d(X^n,\hat{X}_2^n)>D_2 \} \leq\Pr\{d(X^n,\hatX_1^n)\leq |P_Z-D_2|\}+\exp\{-M_2\underline{h}(n,\lambda D_1)\}\label{SOspach:step2}.
\end{align}

Choose $M_2$ such that
\begin{align}
\log M_2&=-\log \underline{h}(n,\lambda D_1)+\log(\log\sqrt n) \label{SOspach:frac_of_Gamma_function}\\
& = \frac{n}{2}\log \frac{\lambda D_1}{D_2} + O(\log n) \label{SOspach:choose_of_M2},
\end{align}
where \eqref{SOspach:choose_of_M2} follows from the definition of $\underline{h}(n,l)$ in \eqref{def:h} and the fact that $P_Z=\lambda D_1-D_2$ and $\Gamma \left( \frac{n+2}{2} \right)/\Gamma \left(\frac{n+1}{2}\right)=O(\sqrt n)$. With the choice of $M_2$ in \eqref{SOspach:frac_of_Gamma_function}, we have
\begin{align}
\exp\{ -M_2 \underline{h}(n,\lambda D_1) \}=\frac{1}{{\sqrt n }} \label{SOspach:solve_second_term_of_second_layer}.
\end{align}

To upper bound the first term in \eqref{SOspach:step2}, similarly to \cite[Section IV. C]{zhou2018refined}, we define the sets
\begin{align}
& \mathcal{P} :=  \{ r\in \mathbb{R}: b < r-\sigma^2 \leq a_n  \} \label{def:mathcal{P}},\\
& \mathcal{Q} :=  \{ r\in \mathbb{R}: r+P_Y-|P_Z-D_2| \geq 0 \}, \label{def:mathcal{Q}}
\end{align}
where the value of $b$ is to be specified.

Recall that $w=\frac{\|x^n\|^2}{n}$ is the  normalized $\ell_2$ norm of $x^n$. Furthermore, we need  the following definition:
\begin{align}
R_{\rm{sp}}(w,P,D)&:=  - \frac{1}{2}\log \left(1-\frac{(w+P-D)^2}{4wP}\right).
\end{align}
Recall the definition of $\Psi_{\mathrm{sp}}(n,w,D)$ in \eqref{def:Psi_sp}. It follows from~\cite[cf. (63)-(68)]{zhou2018refined} that
\begin{align}
\nn\Psi_{\mathrm{sp}}(n,w,|P_Z-D_2|)
&\leq\frac{1}{\sqrt \pi}\frac{\Gamma ( \frac{n}{2})}{\Gamma \left( \frac{n - 1}{2} \right)}\exp \{-(n-3)R_{\mathrm{sp}}(w,P_Y,|P_Z-D_2|)\}\\*
&=:\bar{g}(n,w,|P_Z-D_2|), \label{def:bar_g_prime}
\end{align}
and similar to~\cite[cf. (69)-(77)]{zhou2018refined}, we have
\begin{align}
&\nn\Pr \{d(X^n,\hat{X}_1^n) \geq |P_Z-D_2|\} \\
&\qquad=\int_0^\infty (1-\Psi_{\mathrm{sp}}(n,w,|P_Z-D_2|))^{M_1}f_W(w)\rmd w \\
&\qquad\geq \left( 1 - \frac{1}{\sqrt n } \right)\Pr \biggl\{ W \in {\cal P} \cap {\cal Q},\;\log M_1 \le -\log 2 - \log \bar g(n,\sigma ^2 + b,|P_Z-D_2|)-\frac{1}{2}\log n \biggr\} \label{SOspach:error_prob_of_first_term_in_second_layer}.
\end{align}
We next further lower bound the probability term in \eqref{SOspach:error_prob_of_first_term_in_second_layer} and show that the probability term tends to 1 asymptotically. Note that $R_{\rm{sp}}(w,P,D)$ is a decreasing function of $D$ and is an increasing function of $w$. For $w_1=\sigma^2$ and $\lambda D_1$, we have $ R_{\rm{sp}}(w_1,P_Y,\lambda D_1)=\frac{1}{2}\log \frac{\sigma ^2}{\lambda D_1}$. Recall that $P_Z=\lambda D_1-D_2$, since $|\lambda D_1-2D_2|<\lambda D_1$, if we choose $w_2\in\bbR_+$ such that
\begin{align}
R_{\rm{sp}}(w_2, P_Y, |\lambda D_1-2D_2|)=\frac{1}{2}\log \frac{\sigma ^2}{\lambda D_1}, \label{SOspach:set_w2}
\end{align}
it follows that $w_2<w_1=\sigma^2$. Set $w_3=\frac{w_1+w_2}{2}$. Thus $w_2<w_3<\sigma^2$ and
\begin{align}
R_{\rm{sp}}(w_3, P_Y, |\lambda D_1-2D_2|)
&>R_{\rm{sp}}(w_2, P_Y, |\lambda D_1-2D_2|)\\
&=\frac{1}{2}\log\frac{\sigma^2}{\lambda D_1}\label{SOspach:relationship_between_R_sp_w3_and_D1w1}.
\end{align}
Set the value of $b$ as
\begin{align}
b&:= \frac{w_2-\sigma^2}{2}. \label{SOspach:set_b}
\end{align}

Using the definition of $\bar{g}(n,w)$ in \eqref{def:bar_g_prime}, it follows that for some $\delta>0$,
\begin{align}
\nn&- \log \bar g( n,\sigma ^2 + b,|P_Z-D_2|)\\*
&\quad= nR_{\rm{sp}}(w_3,P_Y,|\lambda D_1-2D_2|)+O(\log n)\label{SOspach:using_Gamma}\\
&\quad=\frac{n}{2}\log\frac{\sigma^2}{\lambda D_1}+n\delta + O(\log n), \label{SOspach:barg_prime_geq_M1}
\end{align}
where \eqref{SOspach:using_Gamma} follows since $\Gamma \left( \frac{n}{2} \right)/\Gamma \left(\frac{n-1}{2}\right)=O(\sqrt n)$ and $P_Z=\lambda D_1-D_2$, and \eqref{SOspach:barg_prime_geq_M1} follows from the result in \eqref{SOspach:relationship_between_R_sp_w3_and_D1w1}.

Combining \eqref{SOspach:chooseM_1} and \eqref{SOspach:barg_prime_geq_M1} and invoking the weak law of large numbers, we conclude that
\begin{align}
\lim_{n\to\infty}\Pr \left\{ \log M_1 \le  - \log 2 - \log \bar g( n,\sigma ^2 + b) - \frac{1}{2}\log n \right\} =1. \label{SOspach:solve_firstterm_secondlayer}
\end{align}
Note that \eqref{SOspach:solve_firstterm_secondlayer} follows since \eqref{SOspach:barg_prime_geq_M1} implies that the dominant term in the right-hand side is strictly larger than $\frac{n}{2}\log\frac{\sigma^2}{\lambda D_1}$, which is the dominant term of $\log M_1$ (cf. \eqref{SOspach:chooseM_1}). Similar to \cite[Lemma 5]{zhou2018refined}, using the weak law of large numbers, the Berry-Esseen theorem~\cite{berry1941accuracy,esseen1942liapounoff} and the definition of $b$, we conclude that
\begin{align}
\Pr \left\{ {W \in\cal P \cap\cal Q} \right\}
&\geq 1- O\left( \frac{1}{\sqrt n} \right)\label{SOspach:lemma5_in_zhou2018}.
\end{align}
Finally, combining \eqref{SOspach:chooseM_1}, \eqref{SOspach:useproperty}, \eqref{SOspach:choose_of_M2}, \eqref{SOspach:solve_second_term_of_second_layer} and \eqref{SOspach:lemma5_in_zhou2018}, we have
\begin{align}
\lim_{n\to\infty}\Pr\{ d(X^n,\hat{X}_1^n)\leq \lambda D_1, \; d(X^n,\hat{X}_2^n)>D_2 \} =0. \label{SOspach:sp_solve_of_second_layer}
\end{align}

The proof of Theorem \ref{theo:second_JEP} when both encoders use spherical codebooks completed by combining \eqref{SOspach:error_decompose}, \eqref{SOspach:solve_first_layer}, \eqref{SOspach:sp_solve_of_second_layer}.

\subsection{When Both Encoders Use i.i.d. Gaussian Codebooks}\label{sec:proof_SOJEP_iid}
We next present the proof when both encoders use i.i.d. Gaussian codebooks, i.e., $(\dagger,\ddagger)=\rm\{iid,iid\}$.
Similar to the case when both encoders use spherical codebooks, the first term in \eqref{SOspach:error_decompose} is similar to the mismatched rate-distortion problem studied in~\cite[Section IV. D]{zhou2018refined}. Choose $M_1$ such that
\begin{align}
\log M_1=\frac{n}{2}\log\frac{\sigma^2}{\lambda D_1} + \sqrt {n\rm{V}(\sigma^2,\zeta)}\rmQ^{-1}(\varepsilon)+O(\log n). \label{SOiidach:chooseM_1}
\end{align}
Similar to~\cite[cf. (90)-(103)]{zhou2018refined}, we have that
\begin{align}
\Pr\{d(X^n,\hat{X}_1^n)>\lambda D_1\}\leq\varepsilon+O\left(\frac{1}{\sqrt{n}}\right).\label{SOiidach:solve_first_layer}
\end{align}

The second term in \eqref{SOspach:error_decompose} is upper bounded as follows.
\begin{align}
\nn&\Pr\{ d(X^n,\hat{X}_1^n)\leq \lambda D_1, \; d(X^n,\hat{X}_2^n)>D_2 \}\\*
&=\int_{\substack{x^n,\by:\\d(x^n,\hatx_1^n)\leq \lambda D_1}}P_X^n(x^n)\prod_{i\in[M_1]}f_{\rm{iid}}(y^n(i)|\bzero^n,P_Y)
\times \Pr\{\forall~j\in[M_2]:~d(x^n,Z^n(f_1(x^n),j))>D_2\}\rmd x^n\rmd \by \label{SOiidach:decompose_second_layer}\\
&=\int_{\substack{x^n,\by:\\d(x^n,\hatx_1^n)\leq \lambda D_1}}P_X^n(x^n)\prod_{i\in[M_1]}f_{\rm{iid}}(y^n(i)|\bzero^n,P_Y)\times (1-\Phi_{\mathrm{iid}}(n,x^n,\hatx_1^n))^{M_2}\rmd x^n\rmd \by\label{SOiidach:useindependent},\\
&=\int_0^{\lambda D_1}f_L(l)(1-\Phi_{\mathrm{iid}}(n,l))^{M_2}\rmd l \label{SOiidach:usestronglarge}\\
&\leq \int_0^{|P_Z-D_2|}f_L(l)\rmd l+\int_{|P_Z-D_2|}^{\lambda D_1}f_L(l)(1-\Phi_{\mathrm{iid}}(n,l))^{M_2}\rmd l \label{SOiidach:usePhi_range}\\
&\leq \int_0^{|P_Z-D_2|}f_L(l)\rmd l+\int_{|P_Z-D_2|}^{\lambda D_1}f_L(l)\exp\{-M_2\Phi_{\mathrm{iid}}(n,l)\}\rmd l
\label{SOiidach:usingexp{-Ma}}\\
&\leq \Pr\{d(X^n,\hatX_1^n)\leq |P_Z-D_2|\}+\exp\{-M_2\Phi_{\mathrm{iid}}(n,\lambda D_1)\}, \label{SOiidach:removeint_byD_1}
\end{align}
where \eqref{SOiidach:useindependent} follows since given $x^n$, $\by$ and thus $\hatx_1^n$, the codewords $(Z^n(\hatx_1^n,1),\ldots,Z^n(\hatx_1^n,M_2))$ are independent and generated by the same distribution $f_{\rm{iid}}(\cdot|\hatx_1^n,\lambda D_1-D_2)$, \eqref{SOiidach:usestronglarge} follows from \eqref{Soiidach:result_strong_large}, \eqref{SOiidach:usePhi_range} follows since $\Phi_{\mathrm{iid}}(n,l)\leq 1$, \eqref{SOiidach:usingexp{-Ma}} follows since $(1-a)^{M} \leq \exp\{-Ma \}$ for any $a\in [0,1)$ and \eqref{SOiidach:removeint_byD_1} since $\Phi_{\mathrm{iid}}(n,l)$ is a decreasing function of $l$ if $l\geq|P_Z-D_2|$.

Choose $M_2$ such that
\begin{align}
\log M_2&=-\log \Phi_{\mathrm{iid}}(n,\lambda D_1)+\log(\log\sqrt n)\\
& = \frac{n}{2}\log \frac{\lambda D_1}{D_2} + O(\log n) \label{SOiidach:choose_of_M2}.
\end{align}
It follows that
\begin{align}
\exp\{ -M_2 \Phi_{\mathrm{iid}}(n,\lambda D_1) \}=\frac{1}{{\sqrt n }} \label{SOiidach:solve_second_term_of_second_layer}.
\end{align}

The analysis of the first term in \eqref{SOiidach:removeint_byD_1} is similar to spherical codebook (cf. \eqref{def:mathcal{P}} to \eqref{SOspach:lemma5_in_zhou2018}) except the following two points: i) replace $\bar g(n,w)$ with $\Psi_{\mathrm{iid}}(n,w,|P_Z-D_2|)$ and ii) replace $\cal P \cap\cal Q$ with $\cal P$. Hence, we have
\begin{align}
\lim_{n\to\infty}\Pr\{d(X^n,\hat{X}_1^n)\leq \lambda D_1, \; d(X^n,\hat{X}_2^n)>D_2 \} =0. \label{SOiidach:solve_second_layer}
\end{align}
The proof of Theorem \ref{theo:second_JEP} when both encoders use i.i.d. Gaussian codebooks is completed by combining \eqref{SOspach:error_decompose}, \eqref{SOiidach:solve_first_layer}, \eqref{SOiidach:solve_second_layer}.

\subsection{When Two Encoders Use Different Codebooks}
We next prove Theorem \ref{theo:second_JEP} when different types of codebooks are used in encoders $f_1$ and $f_2$.
When $(\dagger,\ddagger)=\rm\{iid,sp\}$, to bound the first term in \eqref{SOspach:error_decompose}, the steps are exactly the same as the case $(\dagger,\ddagger)=\{\rm{iid},\rm{iid}\}$ until \eqref{SOiidach:solve_first_layer}. To upper bound the second term in \eqref{SOspach:error_decompose}, the steps are exactly the same as \eqref{SOspach:decompose_second_layer}-\eqref{SOspach:sp_solve_of_second_layer}. When $(\dagger,\ddagger)=\rm\{sp,iid\}$, the proof is exactly the same as the case $(\dagger,\ddagger)=\{\rm{sp},\rm{sp}\}$ until \eqref{SOspach:solve_first_layer}. To upper bound the second term in \eqref{SOspach:error_decompose}, the proof is exactly the same as \eqref{SOiidach:decompose_second_layer}-\eqref{SOiidach:solve_second_layer}.

\section{Proof of Moderate Deviation Asymptotics}\label{sec:proof_Mod}
\subsection{Preliminaries}
Recall the following version of the Chernoff bound~\cite[25, Th. B.4.1]{bouleau93}
\begin{lemma}\label{theo:lemma_mod}
Given an i.i.d. sequence $X^n$, suppose that the cumulant generating function $\Lambda_X^*(\theta)$ is finite for some positive number $\theta$. For any $t>\bbE[X]$,
\begin{align}
\Pr\left\{\frac{1}{n}\sum_{i=1}^nX_i>t\right\}\leq\exp\{-n\Lambda_X^*(t)\},
\end{align}
and for any $t<\bbE(X)$,
\begin{align}
\Pr\left\{\frac{1}{n}\sum_{i=1}^nX_i<t\right\}\leq\exp\{-n\Lambda_X^*(t)\}.
\end{align}
In both cases, $\Lambda_X^*(t)>0$.
\end{lemma}

In other words, the probability decays exponentially fast if the threshold $t$ deviates from the mean by a constant.

\subsection{When Both Encoders Use Spherical Codebook Under JEP}\label{sec:proof_Mod_JEP_sp}
For any $\theta_1>0$, let
\begin{align}
c_n:=2\sigma^2\theta_1\rho_n.
\end{align}
The proof follows from Section \ref{sec:proof_SOJEP_sp} until \eqref{SOspach:step2} with $c_n$ taking place $b_n$. Combining \eqref{SOspach:error_decompose} and \eqref{SOspach:step2}, we have that
\begin{align}
\rmP_{\mathrm{sp},\mathrm{sp}}^n(D_1,D_2|M_1,M_2)\leq \Pr\{d(X^n,\hat{X}_1^n)>\lambda D_1\}+\Pr\{d(X^n,\hatX_1^n)\leq |P_Z-D_2|\}+\exp\{-M_2\underline{h}(n,\lambda D_1)\}. \label{Modspach:error_decompose}
\end{align}

Note that the first term in \eqref{Modspach:error_decompose} is similar to the mismatched rate-distortion problem studied in~\cite[Section V. B]{zhou2018refined} except that the distortion level is replaced from $D_1$ to $\lambda D_1$ and the power of codebook is replaced from $P_Y=\sigma^2-D_1$ to $P_Y=\sigma^2-\lambda D_1$.

Choose $M_1$ such that
\begin{align}
\log M_1 = n\left(\frac{1}{2}\log\frac{\sigma^2}{\lambda D_1}+\theta_1\rho_n\right), \label{Modspach:choose_M1}
\end{align}
Similar to \cite[cf. (107)-(115)]{zhou2018refined}, we have that
\begin{align}
\Pr\{d(X^n,\hat{X}_1^n)>\lambda D_1\} \leq\exp\{-nt_1\}+\exp\left\{-\frac{n\theta_1^2\rho_n^2}{2\rmV(\sigma^2,\zeta)} +o(n\theta_1^2\rho_n^2)\right\} +\exp\{-n(\theta_1\rho_n)^{3/2}+o(\rho_n)\} \label{Modspach:solve_first_layer}
\end{align}
for some $t_1>0$.

To upper bound the third term in \eqref{Modspach:error_decompose}, for any $\theta_2>0$, we choose $M_2$ such that
\begin{align}
\log M_2 =n\left(\frac{1}{2}\log \frac{\lambda D_1}{D_2}+\theta_2\rho_n\right). \label{Modspach:choose_M2}
\end{align}
With the conditions on $\rho_n$ in \eqref{def:zeta_n} and the choice of $M_2$ in \eqref{Modspach:choose_M2}, it follows that
\begin{align}
\exp\{-M_2\underline{h}(n,\lambda D_1)\}=\exp\{-n\theta_2\rho_n+o(\rho_n)\}. \label{Modspach:solve_second_layer}
\end{align}

To upper bound the second term in \eqref{Modspach:error_decompose}, similarly to \cite[Section V. C]{zhou2018refined}, we define the set
\begin{align}
\calP':=\{r\in\bbR:b<r-\sigma^2\leq2c_n\}.
\end{align}

Following the similar proof in \eqref{SOspach:error_prob_of_first_term_in_second_layer}, with $\calP'$ in place of $\calP$, $\exp\{-n(\theta_1\rho_n)^{3/2}\}$ in place of $\frac{1}{\sqrt{n}}$ and $n(\theta_1\rho_n)^{3/2}$ in place of $\frac{1}{2}\log n$, we have that
\begin{align}
&\nn\Pr\{d(X^n,\hatX_1^n)\geq |P_Z-D_2|\}\\
&\qquad\geq \left(1-\exp\{-n(\theta_1\rho_n)^{3/2}\}\}\right)\Pr \biggl\{ W \in {\calP'} \cap {\cal Q},\;\log M_1\le-\log2-\log\bar g(n,\sigma ^2 + b,|P_Z-D_2|)-n(\theta_1\rho_n)^{3/2} \biggr\}. \label{Modspach:P_Z-D_2geq}
\end{align}

Similarly from \eqref{SOspach:set_w2} to \eqref{SOspach:solve_firstterm_secondlayer}, with the same choice of $b$ in \eqref{SOspach:set_b}, we show that
\begin{align}
\lim_{n\to\infty}\Pr \left\{\log M_1\le-\log2-\log\bar g(n,\sigma ^2+b) -n(\theta_1\rho_n)^{3/2}\right\}
=1. \label{Modspach:solve_firstterm_secondlayer}
\end{align}

Invoking Lemma \ref{theo:lemma_mod}, and moderate deviations theorem in \cite[Th. 3.7.1]{dembo2009large}, we have that
\begin{align}
\Pr \biggl\{W\in{\calP'}\cap{\calQ}\biggl\}
&\geq\Pr\{W\in\calP'\}-\Pr\{W\notin\calQ\}\\
&=\Pr\left\{\frac{1}{n}\sum_{i=1}^{n}X_i^2>\sigma^2+b\right\} -\Pr\left\{\frac{1}{n}\sum_{i=1}^{n}X_i^2>\sigma^2+2c_n\right\} -\Pr\left\{\frac{1}{n}\sum_{i=1}^{n}X_i^2<|P_Z-D_2|-P_Y\right\}\\
&\geq\big(1-\exp\{-nt_2\}\big)- \exp\left\{-\frac{4n\theta_1^2\rho_n^2}{2\rmV(\sigma^2,\zeta)}+o(n\rho_n^2)\right\} -\Pr\left\{\frac{1}{n}\sum_{i=1}^{n}X_i^2<\sigma^2-2P_Y\right\}\label{Modspach:notin_calQ}\\
&\geq\big(1-\exp\{-nt_2\}\big)- \exp\left\{-\frac{4n\theta_1^2\rho_n^2}{2\rmV(\sigma^2,\zeta)}+o(n\rho_n^2)\right\} -\exp\{-nt_3\}\label{Modspach:W_in_P'andQ}
\end{align}
for some $t_2>0$ and $t_3>0$.
Combining the results in \eqref{Modspach:error_decompose}-\eqref{Modspach:solve_second_layer}, \eqref{Modspach:P_Z-D_2geq}, \eqref{Modspach:solve_firstterm_secondlayer} and \eqref{Modspach:W_in_P'andQ}, we conclude that
\begin{align}
\liminf_{n\to\infty}-\frac{1}{n\rho_n^2} \log\rmP_{\mathrm{sp},\mathrm{sp}}^n(D_1,D_2|M_1,M_2) \geq\frac{\theta_1^2}{2\rmV(\sigma^2,\zeta)}.
\end{align}

\subsection{When Both Encoders Use the i.i.d. Gaussian codebook Under JEP} \label{sec:proof_Mod_JEP_iid}

The proof is omitted since it is similar to Section \ref{sec:proof_Mod_JEP_sp} except for following two points: i) replace $\bar g(n,w)$ with $\Psi_{\mathrm{iid}}(n,w,|P_Z-D_2|)$, ii) replace $\cal P' \cap\cal Q$ with $\cal P'$.

\subsection{When Two Encoders Use Different Codebooks Under JEP}
We next prove Theorem \ref{theo:moderate_JEP} when different types of codebooks are used by encoders $f_1$ and $f_2$. When $(\dagger,\ddagger)=\rm\{sp,iid\}$, the proof is exactly the same as the case $(\dagger,\ddagger)=\{\rm{sp},\rm{sp}\}$ until \eqref{Modspach:solve_first_layer}. The rest of the proof is exactly the same as the case $(\dagger,\ddagger)=\{\rm{iid},\rm{iid}\}$.
When $(\dagger,\ddagger)=\rm\{iid,sp\}$, to bound the term $\Pr\{d(X^n,\hat{X}_1^n)>\lambda D_1\}$, the steps are exactly the same as the case $(\dagger,\ddagger)=\{\rm{iid},\rm{iid}\}$. The rest of the proof are exactly the same as \eqref{Modspach:choose_M2}-\eqref{Modspach:W_in_P'andQ}.

\subsection{Under SEP}\label{sec:proof_Mod_SEP}
Under ensemble SEP, each layer of encoder and decoder only consider their own performance. Set $\lambda=1$. Thus, $P_Y=\sigma^2-D_1$ and $P_Z=D_1-D_2$. Note that $v_{\dagger}^*$ is exactly the moderate deviations constant in rate-distortion problem~\cite[Section V.B]{zhou2018refined}, thus we only calculate $v_{\ddagger}^*$ in following section.

The excess-distortion probability $\rmP_{\dagger,\ddagger}^n(D_2|M_1,M_2)$ satisfies
\begin{align}
\nn&\rmP_{\dagger,\ddagger}^n(D_2|M_1,M_2)\\*
&\quad=\Pr\{d(X^n,\hat{X}_2^n)>D_2\}\\
&\quad\leq\Pr\{d(X^n,\hat{X}_1^n)>D_1\;\mathrm{or}\;d(X^n,\hat{X}_2^n)>D_2\}\\
&\quad=\rmP_{\dagger,\ddagger}^n(D_1,D_2|M_1,M_2).
\end{align}
Recalling the results in \ref{sec:proof_Mod_JEP_sp}, with the choice of $M_1$ in \eqref{Modspach:choose_M1} and $M_2$ in \eqref{Modspach:choose_M2}, we have that
\begin{align}
\liminf_{n\to\infty}-\frac{1}{n\rho_n^2} \log\rmP_{\dagger,\ddagger}^n\left(D_2|M_1,M_2\right)\geq\frac{\theta_1^2}{2\rmV(\sigma^2,\zeta)}.
\end{align}

\section{Proof of Large Deviation Asymptotics} \label{sec:proof_large}
\subsection{Preliminaries}
Recall the definition of $\Psi_{\mathrm{iid}}(\cdot)$ in \eqref{def:Psi_sp}, $\Phi_{\mathrm{iid}}(\cdot)$ in \eqref{def:Phi_sp} and the fact that $\Phi_{\mathrm{iid}}$ depends on $(x^n,y^n)$ only through their quadratic distortion $l=d(x^n,y^n)$ (cf. \eqref{Soiidach:result_strong_large}). Similarly, we have that $\Psi_{\mathrm{iid}}(n,x^n,D)$ depends on $x^n$ only through its normalized $\ell_2$-norm $w=\frac{\|x^n\|^2}{n}$:
\begin{align}
\Psi_{\mathrm{iid}}(n,w,D)\sim\frac{\exp\{-nR_{\mathrm{iid}}(w,P_Y,D)\}} {s^*(w,P_Y,D)\sqrt{\kappa(w,P_Y,D)}}. \label{largeach:frist_strong_large}
\end{align}

We use the following properties of $R_{\mathrm{iid}}(l,P,D)$ (cf. \eqref{def:s*} and cf. \eqref{def:Riid}) in the proof.
\begin{lemma}\label{theo:lemma_of_iid_large}
The following results hold.
\begin{enumerate}[i)]
\item Given $\lambda\in(\frac{D_2}{D_1},1]$, $R_{\mathrm{iid}}(l,P,D)=0$ if $l=|D-P|^+$, $R_{\mathrm{iid}}(\sigma^2,P_Y, \lambda D_1)= \frac{1}{2}\log\frac{\sigma^2}{\lambda D_1}$ and $R_{\mathrm{iid}}(\lambda D_1,P_Z,D_2)= \frac{1}{2}\log\frac{\lambda D_1}{D_2}$.
\item $s^*(l,P,D)>0$ if and only if $l>|D-P|^+$;
\item $R_{\mathrm{iid}}(l,P,D)=
  \sup_{s\geq0}R_{\mathrm{iid}}(l,P,D)$ and thus $R_{\mathrm{iid}}(l,P,D)$ increases in $l$ when $l\geq|D-P|^+$ and decreases in $D$ when $D\leq l+P$.
\item $\Lambda^*_{X^2}(t)>0$ if and only if $t>\sigma^2=\bbE[X^2]$.
\end{enumerate}
\end{lemma}
The proof of Lemma \ref{theo:lemma_of_iid_large} is similar to~\cite[Lemma 7]{zhou2018refined} and thus omitted.

\subsection{Proof under JEP (Theorem \ref{theo:large_JEP})}\label{sec:proof_large_JEP}
Fix $\lambda\in(\frac{D_2}{D_1},1]$. It follows that $\lambda D_1>|D_2-P_Z|^+$. Using \eqref{Soiidach:result_strong_large}, for any $l\leq\lambda D_1$, any positive $\delta$ and sufficiently large $n$, we have
\begin{align}
\Phi_{\mathrm{iid}}(n,l)&\geq\Phi_{\mathrm{iid}}(n,\lambda D_1)\\
&\geq\exp\{-n(1+\delta)R_{\mathrm{iid}}(\lambda D_1,P_Z,D_2)\}. \label{largeach:i_Phi_expiid}
\end{align}

It follows from the definition of the ensemble JEP $\rmP_{\dagger,\ddagger}^n(D_1,D_2|M_1,M_2)$ in \eqref{def:P_joint} that
\begin{align}
\rmP_{\mathrm{iid},\mathrm{iid}}^n(D_1,D_2|M_1,M_2)
=\Pr\{d(X^n,\hat{X}_1^n)>D_1\} +\Pr\{d(X^n,\hat{X}_1^n)\leq D_1,\;d(X^n,\hat{X}_2^n)>D_2 \}. \label{largeach:i_error_decompose}
\end{align}

For $\lambda\in(\frac{D_2}{D_1},1)$, the second term in \eqref{largeach:i_error_decompose} can be further upper bounded as follows,
\begin{align}
\nn\Pr&\{ d(X^n,\hat{X}_1^n)\leq D_1, \; d(X^n,\hat{X}_2^n)>D_2 \}\\
&=\int_{\substack{x^n,\by:\\d(x^n,\hatx_1^n)\leq D_1}} P_X^n(x^n)\prod_{i\in[M_1]}f_{\rm{iid}}(y^n(i)|\bzero^n,P_Y)\times \Pr\{\forall~j\in[M_2]:~d(x^n,Z^n(f_1(x^n),j))>D_2\}\rmd x^n\rmd \by\\
&=\int_{\substack{x^n,\by:\\d(x^n,\hatx_1^n)\leq D_1}} P_X^n(x^n)\prod_{i\in[M_1]}f_{\rm{iid}}(y^n(i)|\bzero^n,P_Y)
\times (1-\Phi_{\mathrm{iid}}(n,x^n,\hatx_1^n))^{M_2}\rmd x^n\rmd \by \label{largeach:i_usePhiiid}\\
&=\int_{0}^{D_1}f_L(l)(1-\Phi_{\mathrm{iid}}(n,l))^{M_2}\rmd l\label{largeach:i_usel},\\
&\leq\int_{0}^{\lambda D_1}f_L(l)(1-\Phi_{\mathrm{iid}}(n,l))^{M_2}\rmd l +\Pr\{D_1\geq d(X^n,\hatX_1^n)>\lambda D_1\}\label{largeach:i_usePhiiid>0}\\
&\leq\int_{0}^{\lambda D_1}f_L(l)(1- \exp\{-n(1+\delta)R_{\mathrm{iid}}(\lambda D_1,P_Z,D_2)\})^{M_2}\rmd l +\Pr\{D_1\geq d(X^n,\hatX_1^n)>\lambda D_1\} \label{largeach:i_useRiid}\\
&\leq\exp\big\{-M_2\exp\{-n(1+\delta)R_{\mathrm{iid}}(\lambda D_1,P_Z,D_2)\}\big\} +\Pr\{D_1\geq d(X^n,\hatX_1^n)>\lambda D_1\} \label{largeach:i_use(1-a)^M},
\end{align}
where \eqref{largeach:i_usePhiiid} follows from the definition of $\Phi_{\mathrm{iid}}(\cdot)$, \eqref{largeach:i_usel} follows since $\Phi_{\mathrm{iid}}(n,x^n,y^n)$ depends on $x^n$ and $y^n$ only through their  quadratic distortion $l=\frac{1}{n}\|x^n-y^n\|^2$, \eqref{largeach:i_usePhiiid>0} follows since $\Phi_{\mathrm{iid}} (n,x^n,y^n)\geq0$, \eqref{largeach:i_useRiid} follows from \eqref{largeach:i_Phi_expiid} and \eqref{largeach:i_use(1-a)^M} follows since $(1-a)^M\leq\exp\{-Ma\}$ for any $a\in[0,1)$.

Combining the result in \eqref{largeach:i_error_decompose} and \eqref{largeach:i_use(1-a)^M}, we have
\begin{align}
\rmP_{\mathrm{iid},\mathrm{iid}}^n(D_1,D_2|M_1,M_2)\leq \exp\big\{-M_2\exp\{-n(1+\delta)R_{\mathrm{iid}}(\lambda D_1,P_Z,D_2)\}\big\} +\Pr\{d(X^n,\hatX_1^n)>\lambda D_1\}. \label{largeach:i_error_decompose_lambda}
\end{align}

Choose $M_2$ such that
\begin{align}
\log M_2&=n(1+2\delta)R_{\rm{iid}}(\lambda D_1,P_Z,D_2)\\
&=\frac{n(1+2\delta)}{2}\log\frac{\lambda D_1}{D_2}. \label{largeach:i_chooseM_2}
\end{align}

It follows that
\begin{align}
\exp\big\{-M_2\exp\{-n(1+\delta)R_{\mathrm{iid}}(\lambda D_1,P_Z,D_2)\}\big\}=\exp\big\{\exp\{-n\delta R_{\mathrm{iid}}(\lambda D_1,P_Z,D_2)\}\big\}, \label{largeach:i_M_2_doubly}
\end{align}
which vanishes doubly exponentially fast since $\lambda D_1>|D_2-P_Z|^+$ and the Claims i), ii) and iii) of Lemma \ref{theo:lemma_of_iid_large} ensures that $R_{\mathrm{iid}}(\lambda D_1,P_Z,D_2)>0$.

Finally, we bound the second term in \eqref{largeach:i_error_decompose_lambda}. Fix $\alpha_2$ such that $\alpha_2>|\lambda D_1-P_Y|^+$. Using \eqref{largeach:frist_strong_large}, for any $w\leq\alpha_2$, any positive $\delta$ and sufficiently large $n$, we have
\begin{align}
\Psi_{\mathrm{iid}}(n,w,\lambda D_1)&\geq\Psi_{\mathrm{iid}}(n,\alpha_2,\lambda D_1)\\
&\geq\exp\big\{-n(1+\delta)R_{\rm{iid}}(\alpha_2,P_Y,\lambda D_1)\big\}. \label{largeach:i_Psiiid_expiid}
\end{align}

It follows that
\begin{align}
\Pr\{d(X^n,\hatX_1^n)>\lambda D_1\}
&=\bbE_{f_X^n}\big[(1-\Pr\{d(X^n,Y^n)\leq\lambda D_1|X^n\})\big]\\
&=\int_{0}^{\infty}(1-\Psi_{\mathrm{iid}}(n,w,\lambda D_1))^{M_1}f_W(w)dw \label{largeach:i_usedef_Psiiid}\\
&\leq\int_{0}^{\alpha_2}(1-\Psi_{\mathrm{iid}}(n,\alpha_2,\lambda D_1))^{M_1}f_W(w)dw +\Pr\left\{\frac{1}{n}\sum\limits_{i=1}^n X_i^2\geq\alpha_2\right\} \label{largeach:i_usePsiiid>0}\\
&\leq\int_{0}^{\alpha_2}(1- \exp\{-n(1+\delta)R_{\mathrm{iid}} (\alpha_2,P_Y,\lambda D_1)\})^{M_1} f_W(w)dw +\Pr\left\{\frac{1}{n}\sum\limits_{i=1}^n X_i^2\geq\alpha_2\right\} \label{largeach:i_usePsiiid_expiid}\\
&\leq\exp\Big\{-M_1\exp\big\{-n(1+\delta)R_{\mathrm{iid}}(\alpha_2,P_Y,\lambda D_1)\big\}\Big\} +\Pr\left\{\frac{1}{n}\sum\limits_{i=1}^n X_i^2\geq\alpha_2 \right\}, \label{largeach:i_usePsiiid_(1-a)^M}
\end{align}
where \eqref{largeach:i_usedef_Psiiid} follows from the definition of $\Psi_{\mathrm{iid}}(\cdot)$ in \eqref{def:Psi_sp}, \eqref{largeach:i_usePsiiid>0} follows since $(1-\Psi_{\mathrm{iid}}(n,w,\lambda D_1))\leq1$, \eqref{largeach:i_usePsiiid_expiid} follows from \eqref{largeach:i_Psiiid_expiid}, and \eqref{largeach:i_usePsiiid_(1-a)^M} follows since $(1-a)^M\leq\exp\{-Ma\}$ for any $a\in[0,1)$.

Choose $M_1$ such that
\begin{align}
\log M_1=n(1+2\delta)R_{\rm{iid}}(\alpha_2,P_Y,\lambda D_1). \label{largeach:i_chooseM_1}
\end{align}
The first term of \eqref{largeach:i_usePsiiid_(1-a)^M} vanishes doubly exponentially since $\alpha_2>|\lambda D_1-P_Y|^+$ and Lemma \ref{theo:lemma_of_iid_large} implies that  $R_{\rm{iid}}(\alpha_2,P_Y,\lambda D_1)>0$. Using Cram\'er's Theorem \cite[Th. 2.2.3]{dembo2009large} and the definition of $\Lambda_{X^2}^*(\cdot)$, we have
\begin{align}
\Pr\left\{\frac{1}{n}\sum\limits_{i=1}^n X_i^2\geq\alpha_2\right\} \leq\exp\{-n\Lambda_{X^2}^*(\alpha_2)\} \label{largeach:i_useLambda}.
\end{align}
Combining the results in \eqref{largeach:i_error_decompose_lambda}, \eqref{largeach:i_usePsiiid_(1-a)^M} and \eqref{largeach:i_useLambda}, we have
\begin{align}
\rmP_{\mathrm{iid},\mathrm{iid}}^n(D_1,D_2|M_1,M_2)
\nn&\leq \exp\big\{-M_2\exp\{-n(1+\delta)R_{\mathrm{iid}}(\lambda D_1,P_Z,D_2)\}\big\} \\*
&+\exp\big\{-M_1\exp\{-n(1+\delta)R_{\mathrm{iid}}(\alpha_2,P_Y,\lambda D_1)\}\big\}+\exp\{-n\Lambda_{X^2}^*(\alpha_2)\} \label{largeach:i_upperbound_Pe},
\end{align}
where the the first two terms in \eqref{largeach:i_upperbound_Pe} vanish doubly exponentially fast with the choice of $M_2$ in \eqref{largeach:i_chooseM_2} and $M_1$ in \eqref{largeach:i_chooseM_1}.

Recall that $\Lambda_{X^2}^*(t)=0$ if $t\leq\sigma^2$, $R_{\rm{iid}}(\sigma^2,P_Y,\lambda D_1)=\frac{1}{2}\log\frac{\sigma^2}{\lambda D_1}$ and $R_{\rm{iid}}(\lambda D_1,P_Z,D_2)=\frac{1}{2}\log\frac{\lambda D_1}{D_2}$. Letting $\delta\downarrow0$, we conclude that for any $\lambda\in(\frac{D_2}{D_1},1]$,
\begin{align}
E^*(D_1,D_2|R_1,R_2)\geq \Lambda_{X^2}^*(\alpha_2), \label{largeach:i_done}
\end{align}
where $R_1$ is determined from $R_1=R_{\rm{iid}}(\alpha_2,P_Y,\lambda D_1)$ and $R_2$ is determined from $R_2=\frac{1}{2}\log\frac{\lambda D_1}{D_2}$. The proof Theorem \ref{theo:large_JEP} is now completed.

\subsection{Proof under SEP (Theorem \ref{theo:large_SEP})} \label{sec:proof_large_SEP}
We start with the proof of $\rmP_{\mathrm{iid}}^n(D_1|M_1,M_2)$. Fix $\alpha_1>|D_1-P_Y|^+$. Using \eqref{largeach:frist_strong_large}, for any $w\leq\alpha_1$, any positive $\delta$ and sufficiently large $n$, we have
\begin{align}
\Psi_{\mathrm{iid}}(n,w,D_1)&\geq\Psi_{\mathrm{iid}}(n,\alpha_1,D_1)\\
&\geq\exp\{-n(1+\delta)R_{\rm{iid}}(\alpha_1,P_Y,D_1)\}. \label{largeach:Psiiid_expiid}
\end{align}

Recall the definition of $\rmP_{\mathrm{iid}}^n(D_1|M_1,M_2)$ in \eqref{def:P_dagger}. Similar to \cite[cf. (154)-(158)]{zhou2018refined},
\begin{align}
\rmP_{\mathrm{iid}}^n(D_1|M_1,M_2)&=\Pr\{d(X^n,\hatX_1^n)>D_1\}\\
&\leq\exp\Big\{-M_1\exp\big\{-n(1+\delta)R_{\mathrm{iid}}(\alpha_1,P_Y,D_1)\big\}\Big\} +\Pr\left\{\frac{1}{n}\sum\limits_{i=1}^n X_i^2\geq\alpha_1 \right\}. \label{largeach:ii_usePsiiid_(1-a)^M}
\end{align}

Choose $M_1$ such that
\begin{align}
\log M_1=n(1+2\delta)R_{\rm{iid}}(\alpha_,P_Y,D_1). \label{largeach:ii_chooseM_1}
\end{align}
The first term of \eqref{largeach:ii_usePsiiid_(1-a)^M} vanishes doubly exponentially fast which follows from Lemma \ref{theo:lemma_of_iid_large} and the fact that $\alpha_1>|D_1-P_Y|^+$.

Invoking the definition of Cram\'er's Theorem \cite[Th. 2.2.3]{dembo2009large} and the definition of $\Lambda_{X^2}^*(\cdot)$, we obtain that
\begin{align}
\Pr\left\{\frac{1}{n}\sum\limits_{i=1}^n X_i^2\geq\alpha_1\right\} \leq\exp\{-n\Lambda_{X^2}^*(\alpha_1)\} \label{largeach:ii_useLambda}.
\end{align}

Recall that $\Lambda_{X^2}^*(t)=0$ if $t\leq\sigma^2$ and $|D_1-P_Y|^+<\sigma^2$. Letting $\delta\downarrow0$, we conclude that for any $\alpha>|D_1-P_Y|^+$,
\begin{align}
\liminf_{n\to\infty}-\frac{1}{n}\log \rmP_{\rm{iid}}^n(D_1|M_1,M_2)\geq \Lambda_{X^2}^*(\alpha_1),
\end{align}
where $R_1$ is determined by of $R_1=R_{\rm{iid}}(\alpha_1,P_Y,D_1)$.

We next analyze $\rmP_{\dagger}^n(D_2|M_1,M_2)$ when $R_2\leq\frac{1}{2}\log\frac{D_1}{D_2}$. Fix $\lambda\in(\frac{D_2}{D_1},1]$,  it follows that $\lambda D_1>|D_2-P_Z|^+$. Using \eqref{Soiidach:result_strong_large}, for any $l\leq\lambda D_1$, any positive $\delta$ and sufficiently large $n$, we have
\begin{align}
\Phi_{\mathrm{iid}}(n,l)&\geq\Phi_{\mathrm{iid}}(n,\lambda D_1)\\
&\geq\exp\{-n(1+\delta)R_{\mathrm{iid}}(\lambda D_1,P_Z,D_2)\}. \label{largeach:iii_Phi_expiid}
\end{align}

Recall the definition of $\rmP_{\ddagger}^n(D_2|M_1,M_2)$ in \eqref{def:P_ddagger}, we have
\begin{align}
\nn\rmP_{\mathrm{iid}}^n&(D_2|M_1,M_2)\\
&=\Pr\{d(X^n,\hat{X}_2^n)>D_2\}\\ &=\int_{x^n,\by}P_X^n(x^n)\prod_{i\in[M_1]}f_{\rm{iid}}(y^n(i)|\bzero^n,P_Y)\times \Pr\{\forall~j\in[M_2]:~d(x^n,Z^n(f_1(x^n),j))>D_2\}\rmd x^n\rmd \by\\
&=\int_{\substack{x^n,\by}}P_X^n(x^n)\prod_{i\in[M_1]}f_{\rm{iid}}(y^n(i)|\bzero^n,P_Y)
\times (1-\Phi_{\mathrm{iid}}(n,x^n,\hatx_1^n))^{M_2}\rmd x^n\rmd \by\label{largeach:iii_usePhiiid}\\
&=\int_{0}^{\infty}f_L(l)(1-\Phi_{\mathrm{iid}}(n,l))^{M_2}\rmd l\label{largeach:iii_usel}\\
&\leq\int_{0}^{\lambda D_1}f_L(l)(1-\Phi_{\mathrm{iid}}(n,l))^{M_2}\rmd l +\Pr\{d(X^n,\hatX_1^n)>\lambda D_1\}\label{largeach:iii_usePhiiid>0}\\
&\leq\int_{0}^{\lambda D_1}f_L(l)\big(1- \exp\{-n(1+\delta)R_{\mathrm{iid}}(\lambda D_1,P_Z,D_2)\}\big)^{M_2}\rmd l +\Pr\{d(X^n,\hatX_1^n)>\lambda D_1\} \label{largeach:iii_useRiid}\\
&\leq\exp\big\{-M_2\exp\{-n(1+\delta)R_{\mathrm{iid}}(\lambda D_1,P_Z,D_2)\}\big\} +\Pr\{d(X^n,\hatX_1^n)>\lambda D_1\} \label{largeach:iii_use(1-a)^M},
\end{align}
where \eqref{largeach:iii_usePhiiid} follows from the definition of $\Phi_{\mathrm{iid}}(\cdot)$, \eqref{largeach:iii_usel} follows since $\Phi_{\mathrm{iid}}(n,x^n,y^n)$ depends on $x^n$ and $y^n$ only through $l=\frac{1}{n}\|x^n-y^n\|^2$, \eqref{largeach:iii_usePhiiid>0} follows since $\Phi_{\mathrm{iid}} (n,x^n,y^n)\geq0$, \eqref{largeach:iii_useRiid} follows from \eqref{largeach:iii_Phi_expiid} and \eqref{largeach:iii_use(1-a)^M} follows since $(1-a)^M\leq\exp\{-Ma\}$ for any $a\in[0,1)$.

Note that \eqref{largeach:iii_use(1-a)^M} is the same as \eqref{largeach:i_error_decompose_lambda}, the rest of the proof follows from \eqref{largeach:i_chooseM_2} to \eqref{largeach:i_done} and thus omitted.

Finally, we study the case when $R_2>\frac{1}{2}\log\frac{D_1}{D_2}$. Since parameter $\lambda$ in introduced to adapt the situation that $R_1$ is large but $R_2$ is small, we set $\lambda=1$ in this case. Fix $\gamma>D_1$. Using \eqref{Soiidach:result_strong_large}, for any $l\leq\gamma$, any positive $\delta$ and sufficiently large $n$, we have
\begin{align}
\Omega(n,l)&\geq\Omega(n,\gamma_2)\\
&\geq\exp\{-n(1+\delta)R_{\mathrm{iid}}(\gamma,P_Z,D_2)\}. \label{largeach:iiii_Omega_expiid}
\end{align}

Similar to the case where $R_2\leq\frac{1}{2}\log\frac{D_1}{D_2}$, the excess-distortion probability $\rmP_{\ddagger}^n(D_2|M_1,M_2)$ can be upper bounded as follows,
\begin{align}
\rmP_{\mathrm{iid}}^n(D_2|M_1,M_2)
&=\int_{0}^{\infty}f_L(l)(1-\Phi_{\mathrm{iid}}(n,l))^{M_2}\rmd l\label{largeach:iiii_usel}\\
&\leq\int_{0}^{\gamma}f_L(l)(1-\Omega(n,l))^{M_2}\rmd l +\Pr\{d(X^n,\hatX_1^n)>\gamma\}\label{largeach:iiii_useOmega>0}\\
&\leq\int_{0}^{\gamma}f_L(l)(1- \exp\{-n(1+\delta)R_{\mathrm{iid}}(\gamma,P_Z,D_2)\})^{M_2}\rmd l +\Pr\{d(X^n,\hatX_1^n)>\gamma\} \label{largeach:iiii_useRiid}\\
&\leq\exp\big\{-M_2\exp\{-n(1+\delta)R_{\mathrm{iid}}(\gamma,P_Z,D_2)\}\big\} +\Pr\{d(X^n,\hatX_1^n)>\gamma\}. \label{largeach:iiii_use(1-a)^M}
\end{align}

Choose $M_2$ such that
\begin{align}
\log M_2=n(1+2\delta)R_{\rm{iid}}(\gamma,P_Z,D_2). \label{largeach:iiii_chooseM_2}
\end{align}
It follows that
\begin{align}
\exp\big\{-M_2\exp\{-n(1+\delta)R_{\mathrm{iid}}(\gamma,P_Z,D_2)\}\big\} =\exp\{\exp\{-n\delta R_{\mathrm{iid}}(\gamma,P_Z,D_2)\}\},
\end{align}
which vanishes doubly exponentially fast since $\gamma>|D_2-P_Z|^+$ and Lemma \ref{theo:lemma_of_iid_large} ensure that $R_{\mathrm{iid}}(\gamma,P_Z,D_2)>0$.

Note that $\Pr\{d(X^n,\hatX_1^n)>\gamma\}$ is similar to the second term in \eqref{largeach:i_error_decompose_lambda}, except we change the distortion level to from $\lambda D_1$ to $\gamma$.  Fix $\alpha_2$ such that $\alpha_2>|\gamma-P_Y|^+$. Similarly to \eqref{largeach:i_usedef_Psiiid}-\eqref{largeach:i_usePsiiid_(1-a)^M}, we have
\begin{align}
\Pr\{d(X^n,\hatX_1^n)>\gamma\} \leq\exp\big\{-M_1\exp\{-n(1+\delta)R_{\mathrm{iid}}(\alpha_2,P_Y,\gamma)\}\big\} +\Pr\left\{\frac{1}{n}\sum\limits_{i=1}^n X_i^2\geq\alpha_2\right\}. \label{largeach:iiii_useUpsilon_(1-a)^M}
\end{align}

Choose $M_1$ such that
\begin{align}
\log M_1=n(1+2\delta)R_{\rm{iid}}(\alpha_2,P_Y,\gamma). \label{largeach:iiii_chooseM_1}
\end{align}
It follows that the first term of \eqref{largeach:iiii_useUpsilon_(1-a)^M} vanishes doubly exponentially since $\alpha_2>|\gamma-P_Y|^+$ and Lemma \ref{theo:lemma_of_iid_large} implies that  $R_{\rm{iid}}(\alpha_2,P_Y,\gamma)>0$. The rest of the proof is similar to \eqref{largeach:i_useLambda}-\eqref{largeach:i_done} and thus omitted.

\section{Conclusion}\label{sec:conclusion}
We derived achievable refined asymptotics for the mismatched successive refinement problem where one uses Gaussian codebooks and successive minimum Euclidean distance encoding to compress an arbitrary memoryless source satisfying mild moment constraints. We studied the performance of all four combinations of spherical and i.i.d. Gaussian codebooks in both second-order and moderate deviations asymptotics. For large deviations, we only studied the case where both encoders use i.i.d. Gaussian codebooks and showed that our derived exponent under JEP and exponents under SEP are all positive for rate pairs strictly inside the rate-distortion region of a GMS in the matched case.

There are several future research directions. Firstly, one can derive the ensemble converse results of the rate-distortion region and the second-order, moderate and large deviations, especially when either encoder uses a spherical codebook, in order to check whether our achievability results are ensemble tight. Secondly, for a GMS under quadratic distortion measures in the matched successive refinement problem, it is worthwhile to derive exact large deviations asymptotics under both JEP and SEP and check whether our result is optimal when specialized to a GMS. Finally, one could generalize the results in this paper to mismatched settings of other multiterminal lossy source coding problems, e.g., the Kaspi problem~\cite{kaspi1994,zhou2017non,zhou2017kaspishort}, the lossy Gray-Wyner problem~\cite{gray1974source,zhou2015second} and the multiple descriptions problem~\cite{ahlswede1986multiple,ozarow1980source,zhou2017fy}.

\section*{Acknowledgments}
The authors acknowledge two anonymous reviewers for many helpful comments and suggestions, which significantly help improve the clarity and quality of the current manuscript.

\appendix
\subsection{Justification of the Conditions in Case (ii) of Theorem \ref{theo:large_iid_ISIT2022}} \label{subsec:justification_ISIT2022}
We first prove that $R_{\rm{iid}}(|D_2-P_Z|^+,P_Z,D_2)<\frac{1}{2}\log\frac{D_1}{D_2}$ as follows.
\begin{proof}
Recall that $\lambda=1$, $P_Y=\sigma^2-D_1$ and $P_Z=D_1-D_2$. It follows that
\begin{align}
|D_2-P_Z|^+&=\max\{0,2D_2-D_1\}\\
&<D_1, \label{Appendix:2D2-D1}
\end{align}
where \eqref{Appendix:2D2-D1} follows since $D_2<D_1$. Recall that $R_{\mathrm{iid}}(l,P,D)$ increases in $l$ for $l\geq|D-P|^+$ (cf. Lemma \ref{theo:lemma_of_iid_large}, Claim iii)), we have that
\begin{align}
R_{\rm{iid}}(|D_2-P_Z|^+,P_Z,D_2)&<R_{\rm{iid}}(D_1,P_Z,D_2)\\
&=\frac{1}{2}\log\frac{D_1}{D_2}, \label{Appendix:Riid_D1PZD2}
\end{align}
where \eqref{Appendix:Riid_D1PZD2} follows from Claim i), Lemma \ref{theo:lemma_of_iid_large}.
\end{proof}

We next prove that $R_{\rm{iid}}(\max\{\sigma^2,\gamma_2-P_Y\},P_Y,\gamma_2)> \frac{1}{2}\log\frac{\sigma^2}{D_1}$ as follows.
\begin{proof}
Recall that $\gamma_2$ is the solution to $R_2=R_{\rm{iid}}(\gamma_2,P_Z,D_2)$. It follows that
\begin{align}
R_{\rm{iid}}(\gamma_2,P_Z,D_2)&=R_2\\
&<\frac{1}{2}\log\frac{D_1}{D_2} \label{Appendix:conditions_R2}\\
&=R_{\mathrm{iid}}(D_1,P_Z,D_2),
\end{align}
where \eqref{Appendix:conditions_R2} follows from the conditions on $R_2$ in Case (ii) in Theorem \ref{theo:large_iid_ISIT2022}. Thus, we have that $\gamma_2<D_1$ since $R_{\mathrm{iid}}(l,P,D)$ increases in $l$ for $l\geq|D-P|^+$ (cf. Lemma \ref{theo:lemma_of_iid_large}, Claim iii)). Finally, we have that
\begin{align}
R_{\rm{iid}}(\max\{\sigma^2,\gamma_2-P_Y\},P_Y,\gamma_2) &>R_{\rm{iid}}(\sigma^2,P_Y,\gamma_2)\\
&>R_{\mathrm{iid}}(\sigma^2,P_Y,D_1) \label{Appendix:use_gamma2<D1}\\
&=\frac{1}{2}\log\frac{\sigma^2}{D_1}, \label{Appendix:Riid_sigmaPYD1}
\end{align}
where \eqref{Appendix:use_gamma2<D1} follows since $R_{\mathrm{iid}}(l,P,D)$ decreases in $D$ for $D\leq l+P$ (cf. Claim iii) of Lemma \ref{theo:lemma_of_iid_large}) and \eqref{Appendix:Riid_sigmaPYD1} follows from Claim i) of Lemma \ref{theo:lemma_of_iid_large}.
\end{proof}

\subsection{Proof of Lemma \ref{theo:lemma_SO_sp_ach}} \label{subsec:proof_lemma_SO_sp_ach}
Recall the definition of $\Phi_{\mathrm{sp}}(n,x^n,y^n)$ in \eqref{def:Phi_sp}, $\beta_1$ in \eqref{def:beta1} and $\beta_2$ in \eqref{def:beta2}.

\begin{figure}[tb]
\centering
\includegraphics[width=.5\columnwidth]{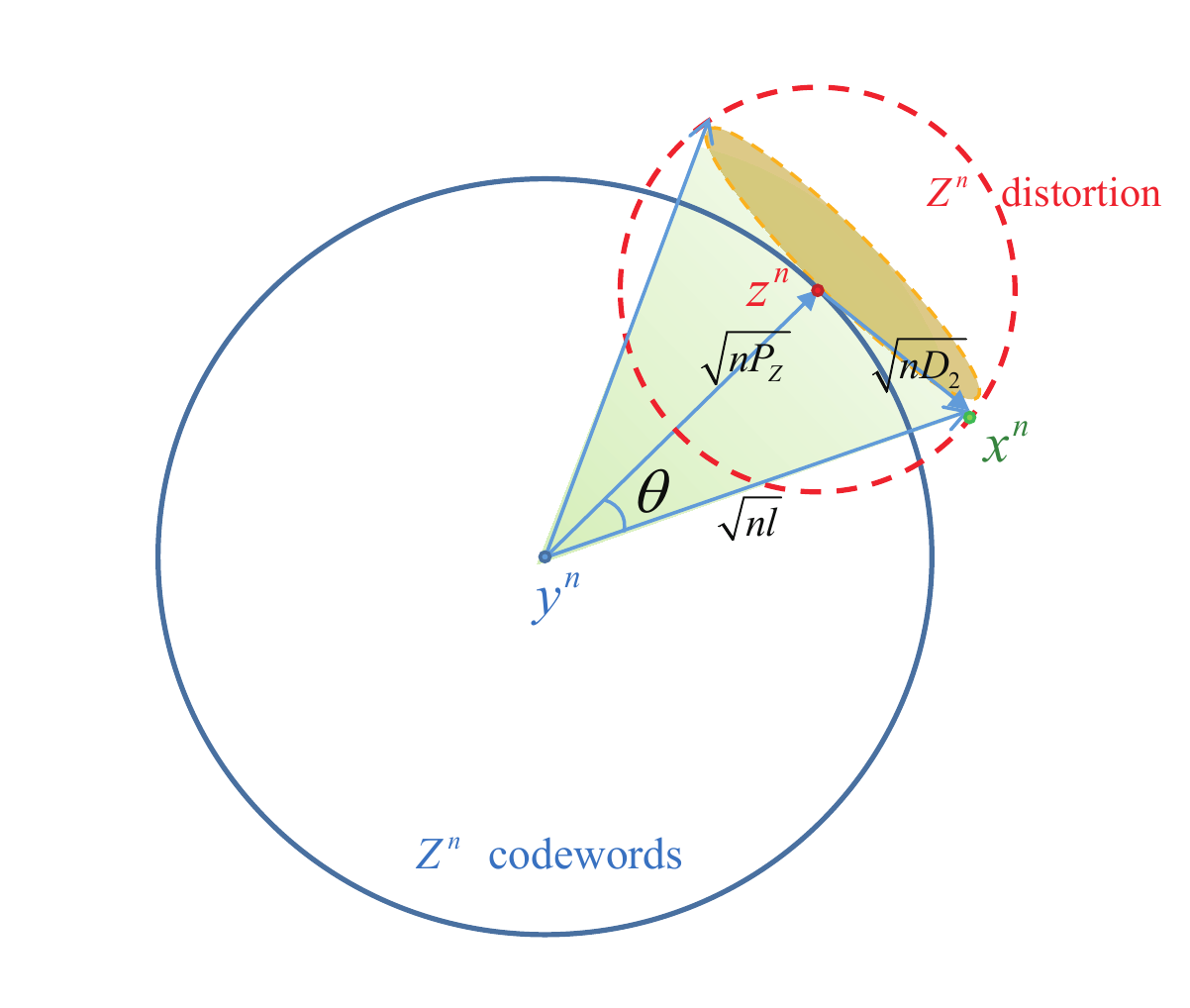}
\caption{Geometry of the excess-distortion probability calculation of encoder $f_2$.}
\label{fig:second_layer}
\end{figure}

Note that $\Phi_{\mathrm{sp}}(n,x^n,y^n)$ depends on $(x^n,y^n)$ only through their quadratic distortion $l=d(x^n,y^n)$. Similarly to~\cite[Theorem 37]{kostina2012}\red{,} $\Phi_{\mathrm{sp}}(n,x^n,y^n)=0$ if $\sqrt{l}\leq|\beta_1|^+$ or $\sqrt{l}\geq\beta_2$. Let $S_n(r) = \frac{n\pi^{\frac{n}{2}}}{\Gamma(\frac{n}{2}+1)}r^{n-1}$ be the surface area of an $n$-dimensional sphere of radius $r$ and $S_n(r,\theta)$ be the surface area of an $n$-dimensional polar cap of radius $r$ and polar angle $\theta$.
For $\sqrt{l}\in[\max \{0,\beta_1\},\beta_2]$, similarly to \cite[Theorem 37]{kostina2012},
\begin{align}
\Phi_{\mathrm{sp}}(n,l)&=\frac{S_n(l,\theta)}{S_n(l)}\\
&\geq\frac{\Gamma(\frac{n}{2}+1)}{\sqrt{\pi}n\Gamma(\frac{n-1}{2}+1)}(\sin\theta)^{n-1} \label{Appendix:S_of_disk}\\
&=\frac{\Gamma(\frac{n+2}{2})}{\sqrt\pi n\Gamma(\frac{n+1}{2})} \left(1-\frac{(l+P_Z-D_2)^2}{4lP_Z}\right)^\frac{n-1}{2}\label{Appendix:law_cos}\\
&=\underline{h}(n,l),
\end{align}
where \eqref{Appendix:S_of_disk} follows since $S_n(r,\theta)\geq\frac{\pi^{\frac{n-1}{2}}}{\Gamma(\frac{n-1}{2}+1)} (r\sin\theta)^{n-1}$ (cf. Fig. \ref{fig:second_layer}) and \eqref{Appendix:law_cos} follows from the law of cosines.

\bibliographystyle{IEEEtran}
\bibliography{IEEEfull_fei}

\begin{thebibliography}{10}
\providecommand{\url}[1]{#1}
\csname url@samestyle\endcsname
\providecommand{\newblock}{\relax}
\providecommand{\bibinfo}[2]{#2}
\providecommand{\BIBentrySTDinterwordspacing}{\spaceskip=0pt\relax}
\providecommand{\BIBentryALTinterwordstretchfactor}{4}
\providecommand{\BIBentryALTinterwordspacing}{\spaceskip=\fontdimen2\font plus
\BIBentryALTinterwordstretchfactor\fontdimen3\font minus
  \fontdimen4\font\relax}
\providecommand{\BIBforeignlanguage}[2]{{%
\expandafter\ifx\csname l@#1\endcsname\relax
\typeout{** WARNING: IEEEtran.bst: No hyphenation pattern has been}%
\typeout{** loaded for the language `#1'. Using the pattern for}%
\typeout{** the default language instead.}%
\else
\language=\csname l@#1\endcsname
\fi
#2}}
\providecommand{\BIBdecl}{\relax}
\BIBdecl

\bibitem{wu2021}
L.~Bai, Z.~Wu, and L.~Zhou, ``Achievable second-order asymptotics for
  successive refinement using {G}aussian codebooks,'' in \emph{IEEE ISIT},
  2021, pp. 2888--2893.

\bibitem{wu2022ISIT}
Z.~Wu, L.~Bai, and L.~Zhou, ``Excess-distortion exponents for successive
  refinement using {G}aussian codebooks,'' in \emph{IEEE ISIT}, 2022, pp.
  234--239.

\bibitem{rimoldi1994}
B.~Rimoldi, ``Successive refinement of information: characterization of the
  achievable rates,'' \emph{IEEE Trans. Inf. Theory}, vol.~40, no.~1, pp.
  253--259, 1994.

\bibitem{kanlis1996error}
A.~Kanlis and P.~Narayan, ``Error exponents for successive refinement by
  partitioning,'' \emph{IEEE Trans. Inf. Theory}, vol.~42, no.~1, pp. 275--282,
  1996.

\bibitem{koshelev1981estimation}
V.~Koshelev, ``Estimation of mean error for a discrete successive-approximation
  scheme,'' \emph{Problemy Pereda\v{c}i Informacii}, vol.~17, no.~3, pp.
  20--33, 1981.

\bibitem{equitz1991successive}
W.~H. Equitz and T.~M. Cover, ``Successive refinement of information,''
  \emph{IEEE Trans. Inf. Theory}, vol.~37, no.~2, pp. 269--275, 1991.

\bibitem{zhou2016second}
L.~Zhou, V.~Y.~F. Tan, and M.~Motani, ``Second-order and moderate deviation
  asymptotics for successive refinement,'' \emph{IEEE Trans. Inf. Theory},
  vol.~63, no.~5, pp. 2896--2921, 2017.

\bibitem{tuncel2003}
E.~Tuncel and K.~Rose, ``Error exponents in scalable source coding,''
  \emph{IEEE Trans. Inf. Theory}, vol.~49, no.~1, pp. 289--296, 2003.

\bibitem{no2016}
A.~No, A.~Ingber, and T.~Weissman, ``Strong successive refinability and
  rate-distortion-complexity tradeoff,'' \emph{IEEE Trans. Inf. Theory},
  vol.~62, no.~6, pp. 3618--3635, 2016.

\bibitem{lapidoth1997}
A.~Lapidoth, ``On the role of mismatch in rate distortion theory,'' \emph{IEEE
  Trans. Inf. Theory}, vol.~43, no.~1, pp. 38--47, 1997.

\bibitem{zhou2018refined}
L.~Zhou, V.~Y. Tan, and M.~Motani, ``Refined asymptotics for rate-distortion
  using {Gaussian} codebooks for arbitrary sources,'' \emph{IEEE Trans. Inf.
  Theory}, vol.~65, no.~5, pp. 3145--3159, 2018.

\bibitem{effros1999}
M.~Effros, ``Distortion-rate bounds for fixed- and variable-rate
  multiresolution source codes,'' \emph{IEEE Trans. Inf. Theory}, vol.~45,
  no.~6, pp. 1887--1910, 1999.

\bibitem{Kostina2019}
V.~Kostina and E.~Tuncel, ``Successive refinement of abstract sources,''
  \emph{IEEE Trans. Inf. Theory}, vol.~65, no.~10, pp. 6385--6398, 2019.

\bibitem{Dembo2002}
A.~Dembo and L.~Kontoyiannis, ``Source coding, large deviations, and
  approximate pattern matching,'' \emph{IEEE Trans. Inf. Theory}, vol.~48,
  no.~6, pp. 1590--1615, 2002.

\bibitem{Kontoyiannis2006}
I.~Kontoyiannis and R.~Zamir, ``Mismatched codebooks and the role of entropy
  coding in lossy data compression,'' \emph{IEEE Trans. Inf. Theory}, vol.~52,
  no.~5, pp. 1922--1938, 2006.

\bibitem{zhou2019jscc}
L.~Zhou, V.~Y.~F. Tan, and M.~Motani, ``The dispersion of mismatched joint
  source-channel coding for arbitrary sources and additive channels,''
  \emph{IEEE Trans. Inf. Theory}, vol.~65, no.~4, pp. 2234--2251, 2019.

\bibitem{kanabar2022mismatched}
M.~Kanabar and J.~Scarlett, ``Mismatched rate-distortion theory: Ensembles,
  bounds, and general alphabets,'' \emph{arXiv:2203.15193}, 2022.

\bibitem{lapidoth1996}
A.~Lapidoth, ``Nearest neighbor decoding for additive {non-Gaussian} noise
  channels,'' \emph{IEEE Trans. Inf. Theory}, vol.~42, no.~5, pp. 1520--1529,
  1996.

\bibitem{scarlett2017mismatch}
J.~Scarlett, V.~Y.~F. Tan, and G.~Durisi, ``The dispersion of nearest-neighbor
  decoding for additive {non-Gaussian} channels,'' \emph{IEEE Trans. Inf.
  Theory}, vol.~63, no.~1, pp. 81--92, 2017.

\bibitem{verger2005covering}
J.-L. Verger-Gaugry, ``Covering a ball with smaller equal balls in
  {$\mathbb{R}^n$},'' \emph{Discrete Comput. Geometry}, vol.~33, no.~1, pp.
  143--155, 2005.

\bibitem{dembo2009large}
A.~Dembo and O.~Zeitouni, \emph{Large Deviations Techniques and
  Applications}.\hskip 1em plus 0.5em minus 0.4em\relax Springer, 2009,
  vol.~38.

\bibitem{ihara2000error}
S.~Ihara and M.~Kubo, ``Error exponent for coding of memoryless {Gaussian}
  sources with a fidelity criterion,'' \emph{IEICE Trans. Fundamentals},
  vol.~83, no.~10, pp. 1891--1897, 2000.

\bibitem{kostina2012}
V.~Kostina and S.~Verd\'u, ``Fixed-length lossy compression in the finite
  blocklength regime,'' \emph{IEEE Trans. Inf. Theory}, vol.~58, no.~6, pp.
  3309--3338, 2012.

\bibitem{berry1941accuracy}
A.~C. Berry, ``The accuracy of the {Gaussian} approximation to the sum of
  independent variates,'' \emph{Trans. Am. Math. Soc.}, vol.~49, no.~1, pp.
  122--136, 1941.

\bibitem{esseen1942liapounoff}
C.-G. Esseen, \emph{On the {Liapounoff} limit of error in the theory of
  probability}.\hskip 1em plus 0.5em minus 0.4em\relax Almqvist \& Wiksell,
  1942.

\bibitem{bouleau93}
N.~Bouleau and D.~L\'epingle, \emph{Numerical Methods for Stochastic
  Processes}.\hskip 1em plus 0.5em minus 0.4em\relax Wiley, 1993.

\bibitem{kaspi1994}
A.~Kaspi, ``Rate-distortion function when side-information may be present at
  the decoder,'' \emph{IEEE Trans. Inf. Theory}, vol.~40, no.~6, pp.
  2031--2034, 1994.

\bibitem{zhou2017non}
L.~Zhou and M.~Motani, ``Non-asymptotic converse bounds and refined asymptotics
  for two source coding problems,'' \emph{IEEE Trans. Inf. Theory}, vol.~65,
  no.~10, pp. 6414--6440, 2019.

\bibitem{zhou2017kaspishort}
------, ``Kaspi problem revisited: Non-asymptotic converse bound and
  second-order asymptotics,'' in \emph{IEEE Globecom}, 2017.

\bibitem{gray1974source}
R.~Gray and A.~Wyner, ``Source coding for a simple network,'' \emph{Bell Syst.
  Tech. J.}, vol.~53, no.~9, pp. 1681--1721, 1974.

\bibitem{zhou2015second}
L.~Zhou, V.~Y.~F. Tan, and M.~Motani, ``Discrete lossy {Gray-Wyner} revisited:
  Second-order asymptotics, large and moderate deviations,'' \emph{IEEE Trans.
  Inf. Theory}, vol.~63, no.~3, pp. 1766--1791, 2017.

\bibitem{ahlswede1986multiple}
R.~Ahlswede, ``On multiple descriptions and team guessing,'' \emph{IEEE Trans.
  Inf. Theory}, vol.~32, no.~4, pp. 543--549, 1986.

\bibitem{ozarow1980source}
L.~Ozarow, ``On a source-coding problem with two channels and three
  receivers,'' \emph{Bell Syst. Tech. J.}, vol.~59, no.~10, pp. 1909--1921,
  1980.

\bibitem{zhou2017fy}
L.~Zhou and M.~Motani, ``On the multiple description coding with one
  semi-deterministic distortion measure,'' in \emph{IEEE Globecom}, 2017.

\end{thebibliography}
\end{document}